\documentclass[12pt]{article}
\usepackage{amsmath}
\usepackage{graphicx}
\usepackage{enumerate}
\usepackage{natbib}
\usepackage{url} 

\newcommand{\blind}{1}

\usepackage{amssymb}
\newcommand{\indep}{\perp \!\!\! \perp}

\usepackage{tikz}
\usetikzlibrary{shapes.geometric, arrows}
\tikzstyle{node.box} = [rectangle, rounded corners, minimum width=3cm, minimum height=1cm, align=center, draw=black]
\tikzstyle{node.text} = [minimum width=3cm, minimum height=1cm, align=center]
\tikzstyle{arrow} = [thick,->,>=stealth]

\usepackage{amsthm}
\newtheorem{theorem}{Theorem}
\newtheorem{proposition}{Proposition}
\newtheorem{corollary}{Corollary}
\newtheorem{lemma}{Lemma}
\newtheorem{assumption}{Assumption}
\newtheorem{condition}{Condition}

\theoremstyle{remark}
\newtheorem{remark}{Remark}

\usepackage[colorlinks=true,citecolor=blue,urlcolor=red]{hyperref}
\usepackage{algorithm}

\begin{document}

\def\spacingset#1{\renewcommand{\baselinestretch}%
{#1}\small\normalsize} \spacingset{1}


\if1\blind
{
  \title{\bf Efficient and robust transfer learning of optimal individualized treatment regimes with right-censored survival data}
  \author{Pan Zhao\thanks{Email: \href{mailto:pan.zhao@inria.fr}{pan.zhao@inria.fr}} , Julie Josse \thanks{Email: \href{mailto:julie.josse@inria.fr}{julie.josse@inria.fr}}\hspace{.2cm}\\
    PreMeDICaL, Inria-Inserm, Montpellier, France\\
    and \\
    Shu Yang\thanks{Email: \href{mailto:syang24@ncsu.edu}{syang24@ncsu.edu}} \\
    Department of Statistics, North Carolina State University, U.S.A.}
  \maketitle
} \fi

\if0\blind
{
  \bigskip
  \bigskip
  \bigskip
  \begin{center}
    {\LARGE\bf Efficient and robust transfer learning of optimal individualized treatment regimes with right-censored survival data}
\end{center}
  \medskip
} \fi

\bigskip
\begin{abstract}
An individualized treatment regime (ITR) is a decision rule that assigns treatments based on patients' characteristics. The value function of an ITR is the expected outcome in a counterfactual world had this ITR been implemented. Recently, there has been increasing interest in combining heterogeneous data sources, such as leveraging the complementary features of randomized controlled trial (RCT) data and a large observational study (OS). Usually, a covariate shift exists between the source and target population, rendering the source-optimal ITR unnecessarily optimal for the target population. We present an efficient and robust transfer learning framework for estimating the optimal ITR with right-censored survival data that generalizes well to the target population. The value function accommodates a broad class of functionals of survival distributions, including survival probabilities and restrictive mean survival times (RMSTs). We propose a doubly robust estimator of the value function, and the optimal ITR is learned by maximizing the value function within a pre-specified class of ITRs. We establish the $N^{-1/3}$ rate of convergence for the estimated parameter indexing the optimal ITR, and show that the proposed optimal value estimator is consistent and asymptotically normal even with flexible machine learning methods for nuisance parameter estimation. We evaluate the empirical performance of the proposed method by simulation studies and a real data application of sodium bicarbonate therapy for patients with severe metabolic acidaemia in the intensive care unit (ICU), combining a RCT and an observational study with heterogeneity.
\end{abstract}

\noindent%
{\it Keywords:} Policy learning, Semiparametric theory, Covariate shift, Transportability, Data integration
\vfill

\newpage
\spacingset{1.9} 
\section{Introduction}
\label{sec:intro}

Data-driven individualized decision making has recently received increasing interest in many fields, such as precision medicine \citep{kosorok2019precision,tsiatis2019dynamic}, mobile health \citep{trella2022designing}, precision public health \citep{rasmussen2020precision} and econometrics \citep{athey2021policy}. The goal of optimal ITR estimation is to learn a decision rule that assigns the best treatment among possible options to each patient based on their individual characteristics in order to optimize some functional of the counterfactual outcome distribution in the population of interest, also known as the value function. The optimal ITR is the one with the maximal value function, and the value function of the optimal ITR is the optimal value function.

For completely observed data without censoring, one prevailing line of work in the statistical and biomedical literature uses model-based methods to solve the optimal ITR problem, such as Q-learning \citep{robins2004optimal,qian2011performance,laber2014interactive} and A-learning \citep{murphy2003optimal,schulte2014q,shi2018high}. Alternatively, direct model-free or policy search methods have been proposed recently, including the classification perspective \citep{zhang2012estimating,zhang2012robust,zhao2012estimating,rubin2012statistical} and interpretable tree or list-based ITRs \citep{laber2015tree,zhang2015using,zhang2018interpretable}, among others. In clinical studies, right-censored survival data are frequently observed as primary outcomes. Recent extensions of optimal ITR with survival data have been established in \cite{goldberg2012q,cui2017tree,jiang2017estimation,bai2017optimal,diaz2018targeted,zhou2022transformation}.

Researchers have investigated using machine learning algorithms to estimate the optimal ITR from large classes, which cannot be indexed by a finite-dimensional parameter \citep{luedtke2016statistical,luedtke2016super}. One typical instance is that the optimal ITR can be learned from the blip function, which is defined as the additive effect of a blip in treatment on a counterfactual outcome, conditional on baseline covariates \citep{robins2004optimal}; and most existing regression or supervised learning methods can be directly applied \citep{kunzel2019metalearners}. However, the ITRs learned by machine learning methods can be too complex to inform policy-making and clinical practice; to facilitate the integration of data-driven ITRs into practice, it is crucial that estimated ITRs be interpretable and parsimonious \citep{zhang2015using}.

Recently, there has been increasing interest in combining heterogeneous data sources, such as leveraging the complementary features of RCT data and a large OS. For example, in biomedical studies and policy research, RCTs are deemed as the gold standard for treatment effects evaluation. However, due to inclusion or exclusion criteria, data availability, and study design, the enrolled participants in RCT who form the source sample may have systematically different characteristics from the target population. Therefore, findings from RCTs cannot be directly extended to the target population of interest \citep{cole2010generalizing,dahabreh2019extending}. See also \citet{colnet2020causal} and \citet{degtiar2021review} for detailed reviews. Heterogeneity in the populations is of great relevance, and a \emph{covariate shift} usually exists where the covariate distributions differ between the source and target populations; thus, the optimal ITR for the source population is not necessarily optimal for the target population. \cite{zhao2019robustifying} uses data from a single trial study and proposes a two-stage procedure to derive a robust and parsimonious rule for the target population; \cite{mo2021learning} proposes a distributionally robust framework that maximizes the worst-case value function under a set of distributions that are ``close" to the training distribution; \cite{kallus2021more} tackles the lack of overlap for different actions in policy learning based on retargeting; \cite{wu2021transfer} and \cite{chu2022targeted} develop a calibration weighting framework that tailors a targeted optimal ITR by leveraging the individual covariate data or summary statistics from a target population; \cite{sahoo2022learning} uses distributionally robust optimization and sensitivity analysis tools to learn a decision rule that minimizes the worst-case risk incurred under a family of test distributions. However, these methods focus on continuous or binary outcomes and only consider a single sample for worst-case risk minimization; the extension to right-censored survival outcomes within the data integration context has not been studied.

In this paper, we propose a new transfer learning method of finding an optimal ITR from a restricted ITR class under the super population framework where the source sample is subject to selection bias and the target sample is representative of the target population with a known sampling mechanism. Specifically, in our value search method, the value function accommodates a broad class of functionals of survival distributions, including survival probabilities and RMSTs. We characterize the efficient influence function (EIF) of the value function and propose the augmented estimator, which involves models for the survival outcome, propensity score, censoring and sampling processes. The proposed estimator is doubly robust in the sense that it is consistent if either the survival outcome model or the models of the propensity score, censoring, and sampling are correctly specified and is locally efficient when all models are correct. We also consider flexible data-adaptive machine learning algorithms to estimate the nuisance parameters and use the cross-fitting procedure to draw valid inferences under mild regularity conditions and a certain rate of convergence conditions. As we consider a restricted class of ITRs indexed by a Euclidean parameter $\eta$, we also establish the $N^{-1/3}$ convergence rate of $\hat{\eta}$, even though its resultant limiting distribution is not standard, and thus very challenging to characterize. Based on this rate of convergence, we show that the proposed estimator for the target value function is consistent and asymptotically normal, even with flexible machine learning methods for nuisance parameter estimation. Interestingly, when the covariate distributions of the source and target populations are the same, i.e., no covariate shift, the semiparametric efficiency bounds of our method and the standard doubly robust method \citep{bai2017optimal} are equal. Moreover, if the true optimal ITR belongs to the restricted class of ITRs, the standard doubly robust method can still learn the optimal ITR despite the covariate shift, but only our method provides valid statistical inference for the value function.

The rest of our paper is organized as follows. In Section~\ref{sec:statfram}, we introduce the statistical framework of causal survival analysis and transfer learning of optimal ITR. Section~\ref{sec:meth} develops the main methodology of learning the value function and associated optimal ITR. Section~\ref{sec:asym} establishes the asymptotic properties of the proposed value estimator. Extensive simulations are reported in Section~\ref{sec:simu} to demonstrate the empirical performance of the proposed method, followed by a real data application given in Section~\ref{sec:rdat}. The article concludes in Section~\ref{sec:disc} with a discussion of some remarks and future work. All proofs and additional results are provided in the Supplementary Material. 

\section{Statistical Framework}
\label{sec:statfram}

\subsection{Causal survival analysis}
\label{subs:csa}

Let $X$ denote the $p$-dimensional vector of covariates that belongs to a covariate space $\mathcal{X} \subset \mathbb{R}^p$, $A \in \mathcal{A} = \{0, 1\}$ denote the binary treatment, and $T \in \mathbb{R}^+$ denote the \emph{survival time} to the event of interest. In the presence of right censoring, the outcome $T$ may not be observed. Let $C \in \mathbb{R}^+$ denote the censoring time and $\Delta = I\{T \leq C\}$ where $I\{\cdot\}$ is the indicator function. Let $U = \min\{T, C\}$ be the observed outcome, $N(t) = I\{U \leq t, \Delta = 1\}$ the counting process, and $Y(t) = I\{U \geq t\}$ the at-risk process.

We use the potential outcomes framework \citep{neyman1923applications,rubin1974estimating}, where for $a \in \mathcal{A} = \{0, 1\}$, $T(a)$ is the survival time had the subject received treatment $a$. The common goal in causal survival analysis is to identify and estimate the counterfactual quantity $\mathbb{E}[y(T(a))]$ for some deterministic transformation function $y(\cdot)$. Such transformations include $y(T) = \min(T, L)$ for the RMST with some pre-specified maximal time horizon $L$, and $y(T) = I\{T \geq t\}$ for the survival probability at time $t$.

Under the standard assumptions (a) consistency: $T = T(A)$, (b) positivity: $Pr(A = a \,|\, X) > 0$ for every $a \in \mathcal{A}$ \emph{almost surely}, (c) unconfoundedness: $A \indep \{T(1), T(0)\} \,|\, X$, (d) conditionally independent censoring: $C \indep \{T(1), T(0)\} \,|\, \{X, A\}$, we can nonparametrically identify $\mathbb{E}[y(T(a))]$ by the outcome regression (OR) formula or the inverse probability weighting (IPW) formula \citep{van2003unified}.

\subsection{ITR and value function}

Without loss of generality, we assume that larger values of $T$ are more desirable. Typically we aim to identify and estimate an ITR $d(x):\mathcal{X} \rightarrow \mathcal{A}$, which is a mapping from the covariate space $\mathcal{X}$ to the treatment space $\mathcal{A} = \{0, 1\}$, that maximizes the expected outcome in a counterfactual world had this ITR been implemented. Suppose $\mathcal{D}$ is the class of candidate ITRs of interest, then define the potential outcome $T(d)$ under any $d \in \mathcal{D}$ by $T(d) = d(X) T(1) + (1 - d(X)) T(0)$, and the value function \citep{manski2004statistical} of $d$ is defined by $V(d) = \mathbb{E}[y(T(d))]$. Then by maximizing $V(d)$ over $\mathcal{D}$, the optimal ITR is defined by $d^{\text{opt}} = \arg\max_{d \in \mathcal{D}} V(d)$. See \cite{qian2011performance} for more details. 

To estimate the value function, we can use the OR or IPW formulas, and also a doubly robust method \citep{bai2017optimal}:
\begin{equation}\label{eq:dr.orig}
\begin{split}
V_{DR} (d) = & \mathbb{E} \bigg[ \frac{I\{A = d(X)\}\Delta\, y(U)}{Pr(A = d(X) \,|\, X) S_C(U \,|\, A, X)} \\
& \quad + \left(1 - \frac{I\{A = d(X)\}}{Pr(A = d(X) \,|\, X)} \right) \mathbb{E}[y(T) \,|\, A= d(X), X] \\
& \quad + \frac{I\{A = d(X)\}}{Pr(A = d(X) \,|\, X)} \int_0^\infty \frac{\mathrm{d}M_C(u \,|\, A, X)}{S_C(u \,|\, A, X)} \mathbb{E}[y(T) \,|\, T \geq u, A, X] \bigg],
\end{split}
\end{equation}
where $S_C(t\,|\,a,x) = Pr(C>t\,|\,A = a, X = x)$ is the conditional survival function for the censoring process, $\mathrm{d} M_C(u \,|\, A = a, X) = \mathrm{d}N_C(u) - Y(u)\mathrm{d}\Lambda_C(u\,|\,A=a,X)$ is the martingale increment for the censoring process, $N_C(u) = I\{U \leq u, \Delta = 0\}$ and $\Lambda_C(u\,|\,A=a,X) = -\log(S_C(u \,|\, A = a, X))$. The first term in \eqref{eq:dr.orig} is the IPW formula, and the augmentation terms capture additional information from the subjects who do not receive treatment $d$, and who receive treatment $d$ but are censored.

In (clinical) practice, it is usually desirable to consider a class of ITRs indexed by a Euclidean parameter $\eta = (\eta_1,\ldots,\eta_{p+1})^T \in \mathbb{R}^{p+1}$ for feasibility and interpretability. Let $V(\eta) = V(d_\eta)$. Throughout, we focus on such a class of linear ITRs:
\begin{equation*}
    \mathcal{D}_\eta = \{d_\eta: d_\eta(X) = I\{\eta^T \tilde{X} \geq 0\}, |\eta_{p+1}| = 1\},
\end{equation*}
where $\tilde{X} = (1, X^T)^T$, and for identifiability we assume there exists a continuous covariate whose coefficient has absolute value one \citep{zhou2022transformation}; without loss of generality, we assume $|\eta_{p+1}| = 1$. Therefore, the population parameter $\eta^\ast$ indexing the optimal ITR is $\eta^\ast = \arg\max_{\eta \in \{\eta \in \mathbb{R}^{p + 1}: |\eta_{p+1}| = 1\}} V(\eta)$, and the optimal value function is $V(\eta^\ast)$.

\subsection{Transfer learning}
The performance of such a learned ITR may suffer from a covariate shift in which the population distributions differ \citep{sugiyama2012machine}. Instead of minimizing the worst-case risk, here we consider a super population framework. Suppose that a source sample of size $n$ and a target sample of size $m$ are sampled independently from the target super population with different mechanisms. Let $I_S$ and $I_T$ denote the indicator of sampling from source and target populations, respectively. A covariate shift means that $Pr(I_S = 1 \,|\, X) \neq Pr(I_T = 1 \,|\, X)$. In the source sample, independent and identically distributed (i.i.d.) data $\mathcal{O}_s = \{X_i, A_i, U_i, \Delta_i, I_{S,i} = 1, I_{T,i} = 0\}_{i=1}^n$ are observed from $n$ subjects; in the target sample, it is common that only the covariates information is available, so i.i.d. data $\mathcal{O}_t = \{X_i, I_{S,i} = 0, I_{T,i} = 1\}_{i=n+1}^{n+m}$ are observed from $m$ subjects. The sampling mechanism and data structure are illustrated in Figure~\ref{fig:data_structure}.

\begin{figure}[ht]
\centering
\caption{Schematic of the data structure of the source and target samples within the target super population framework.}
\begin{tikzpicture}[node distance=2cm]
\node (super_target_popu) [node.box] {Target super population};

\node (finite_source_popu) [node.box, below of=super_target_popu, xshift=-4cm] {Finite population $\{T(1), T(0), X\}$};
\node (finite_target_popu) [node.box, below of=super_target_popu, xshift=4cm] {Finite population $\{T(1), T(0), X\}$};

\node (source_sampling) [node.text, below of=finite_source_popu] {Source sampling $I_S$};
\node (target_sampling) [node.text, below of=finite_target_popu] {Target sampling $I_T$};

\node (source_sample) [node.box, below of=source_sampling] {Complete source sample \\ $\left\{T_i(1), T_i(0), X_i, I_{S,i} = 1, I_{T,i} = 0\right\}_{i=1}^{n}$};
\node (target_sample) [node.box, below of=target_sampling] {Complete target sample \\ $\left\{T_i(1), T_i(0), X_i, I_{S,i} = 0, I_{T,i} = 1\right\}_{i=n + 1}^{n + m}$};

\node (source_mecha) [node.text, below of=source_sample] {Treatment assignment $A$ \\ Censoring $C$};
\node (target_mecha) [node.text, below of=target_sample] {Only observe covariates $X$};

\node (observed_source_sample) [node.box, below of=source_mecha] {Observed source sample \\ $\left\{X_i, A_i, U_i, \Delta_i, I_{S,i} = 1, I_{T,i} = 0\right\}_{i=1}^{n}$};
\node (observed_target_sample) [node.box, below of=target_mecha] {Observed target sample \\ $\left\{X_i, I_{S,i} = 0, I_{T,i} = 1\right\}_{i=n + 1}^{n + m}$};

\draw [arrow] (super_target_popu) -- (finite_source_popu);
\draw [arrow] (super_target_popu) -- (finite_target_popu);
\draw [arrow] (finite_source_popu) -- (source_sampling);
\draw [arrow] (finite_target_popu) -- (target_sampling);
\draw [arrow] (source_sampling) -- (source_sample);
\draw [arrow] (target_sampling) -- (target_sample);
\draw [arrow] (source_sample) -- (source_mecha);
\draw [arrow] (target_sample) -- (target_mecha);
\draw [arrow] (source_mecha) -- (observed_source_sample);
\draw [arrow] (target_mecha) -- (observed_target_sample);

\end{tikzpicture}
\label{fig:data_structure}
\end{figure}
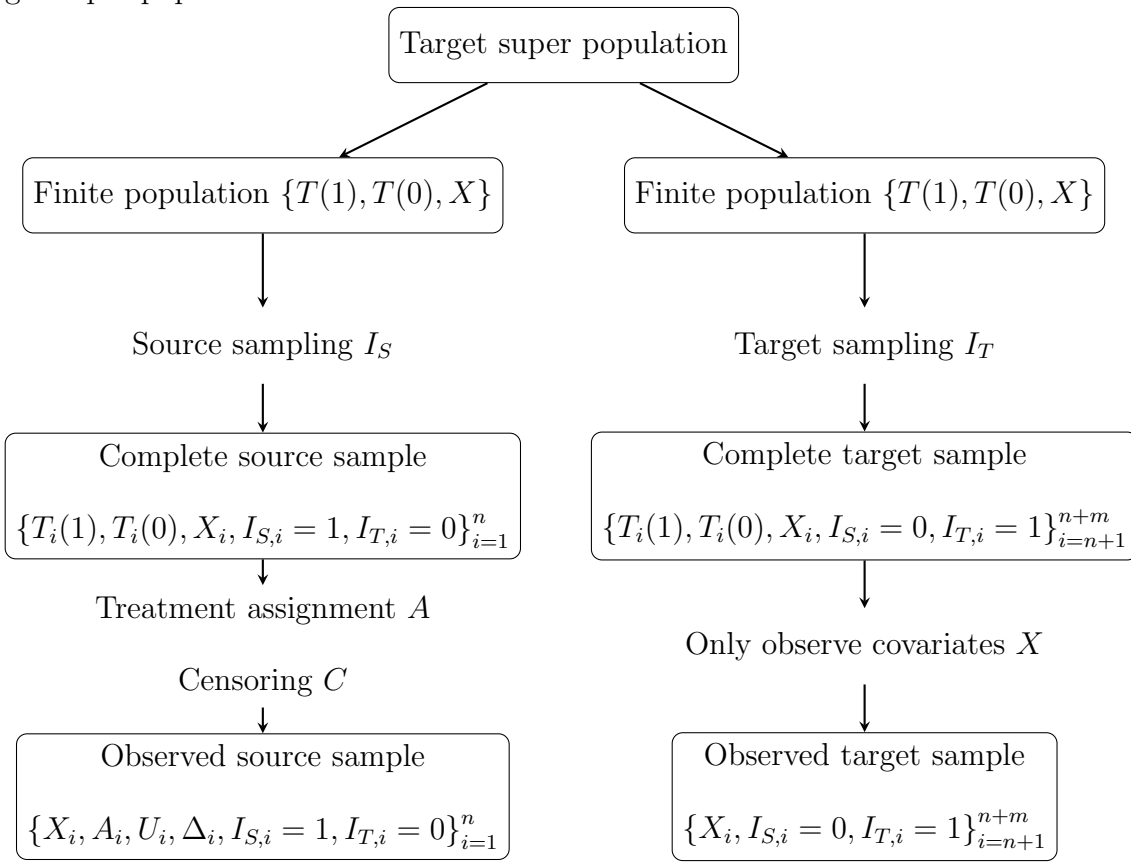

In this framework, we assume that the source and target sampling mechanisms are independent, which holds if two separate studies are conducted independently by different research projects in different locations or in two separate time periods, and the target population is sufficiently large. In the context of combining the RCT and observational study, this framework corresponds to the \emph{non-nested} study design \citep{dahabreh2021study}.

\begin{remark}
In the framework illustrated in Figure~\ref{fig:data_structure}, we also assume the existence of the finite population of size $N$, which helps us clarify the sampling mechanism and identification strategy. The two separate finite populations exemplify the independence of the source and target sampling processes. We present the identification formulas in Section~\ref{sec:meth}; however, we do not require $N$ to be fixed and known. Equivalently, it is also possible to assume a pooled population consisting of a source population and a target population, and similar identification formulas can be proposed based on the density ratio of the two populations.
\end{remark}

\section{Methodology}
\label{sec:meth}

\subsection{Identification and semiparametric efficiency}

To identify the causal effects from the observed data, we make the following assumptions.

\begin{assumption}\label{asmp:cnpc}
(a) $T = T(A)$ almost surely. (b) $Pr(A = a \,|\, X, I_S = 1) > 0$ for every $a$ almost surely. (c) $A \indep \{T(1), T(0)\} \,|\, \{X, I_S = 1\}$. (d) $C \indep \{T(1), T(0)\} \,|\, \{X, A, I_S = 1\}$.
\end{assumption}

Assumption~\ref{asmp:cnpc} includes the standard assumptions as we have introduced in Section~\ref{subs:csa}. Here we only assume them in the source population. Assumption ~\ref{asmp:cnpc}\textit{(a)} implies that the observed outcome is the potential outcome under the actual assigned treatment. Assumption ~\ref{asmp:cnpc}\textit{(b)} states that each subject has a positive probability of receiving both treatments. Assumption ~\ref{asmp:cnpc}\textit{(c)} requires that all confounding factors are measured so that treatment assignment is as good as random conditionally on $X$. Assumption ~\ref{asmp:cnpc}\textit{(d)} essentially states that the censoring process is non-informative conditionally on $X$. Furthermore, we require additional assumptions for the source and target populations.

\begin{assumption}[Survival mean exchangeability]\label{asmp:sme}
$\mathbb{E}[y(T(a)) \,|\, X, I_S = 1] = \mathbb{E}[y(T(a)) \,|\, X]$ for every $a \in \mathcal{A}$.
\end{assumption}

\begin{assumption}[Positivity of Source Inclusion]\label{asmp:psi}
$0 < Pr(I_S = 1 \,|\, X) < 1$ almost surely.
\end{assumption}

\begin{assumption}[Known target design]\label{asmp:ktd}
The target sample design weight $e(x) = \pi_T^{-1}(x) = 1 / Pr(I_T = 1 \,|\, X = x)$ is known by design.
\end{assumption}

Assumption \ref{asmp:sme} is similar to the mean exchangeability over trial participation \citep{dahabreh2019generalizing}, and is weaker than the ignorablility assumption \citep{stuart2011use}, i.e., $I_S \indep \{T(1), T(0)\} \,|\, X$. Assumption \ref{asmp:psi} states that each subject has a positive probability to be included in the source sample, and implies adequate \emph{overlap} of covariate distributions between the source and target populations. Assumption \ref{asmp:ktd} is commonly assumed in the survey sampling literature; thus the design-weighted target sample is representative of the target population. In an observational study with simple random sampling, we have $e(x) = N / m$, where $N$ is the target population size.

Under this framework, we have the following key identity that for any $g(X)$ 
\begin{equation}\label{eq:cw.iden}
\mathbb{E} \left[\frac{I_S}{\pi_S (X)} g(X) \right] = \mathbb{E} [I_T \, e(X) g(X)] = \mathbb{E} [g(X)],
\end{equation}
where $\pi_S(X) = Pr(I_S = 1 \,|\, X)$ is the sampling score.

\begin{proposition}[Identification formulas]\label{prop:iden}
Under Assumptions \ref{asmp:cnpc} - \ref{asmp:ktd}, the value function $V(d)$ can be identified by the outcome regression formula: 
\begin{equation}\label{eq:or}
V(d) = \mathbb{E}[I_T \,e(X)\mathbb{E}[y(T) \,|\, A = d(X), X, I_S = 1]],
\end{equation}
and the IPW formula:
\begin{equation}\label{eq:ipw}
V(d) = \mathbb{E}\left[\frac{I_S}{\pi_S(X)} \frac{I\{A = d(X)\}}{\pi_d(X)} \frac{\Delta\, y(U)}{S_C(U\,|\,A,X)}\right],
\end{equation}
where $\pi_d (X) = d(X) \pi_A(X) + (1 - d(X)) (1 - \pi_A(X))$ with the propensity score $\pi_A(X) = Pr(A = 1 \,|\, X, I_S = 1)$, and $S_C(t \,|\, a, x) = Pr(C>t\,|\,A = a, X = x, I_S = 1)$.
\end{proposition}

Based on the identification formulas~\eqref{eq:or} and \eqref{eq:ipw}, we can construct plug-in estimators for $V(d)$, using the sampling score $\pi_S(X)$ or design weights $e(X)$ to account for the sampling bias. By the identity~\eqref{eq:cw.iden}, the design weights $I_T \, e(X)$ in the OR formula~\eqref{eq:or} with the target sample can also be replaced by the inverse of sampling score $I_S / \pi_{S}(X)$ using the source sample. However, these estimators are biased if the posited models are misspecified, and extreme weights from $\pi_S, \pi_A$ and $S_C$ usually lead to large variability. Therefore, we consider a more efficient and robust approach, motivated by the efficient influence function for $V(d)$.

\begin{proposition}\label{prop:eif}
Under Assumptions \ref{asmp:cnpc} - \ref{asmp:ktd}, the efficient influence function of $V(d)$ is
\begin{align}\label{eq:eif}
\begin{split}
\phi_d = & \frac{I_S}{\pi_S(X)} \frac{I\{A = d(X)\}}{\pi_d(X)} \frac{\Delta\, y(U)}{S_C(U \,|\, A,X)} - V(d) \\ 
& + \left( I_T \,e(X) - \frac{I_S}{\pi_S(X)} \frac{I\{A = d(X)\}}{\pi_d(X)}\right) \mu(d(X), X) \\
& + \frac{I_S}{\pi_S(X)} \frac{I\{A = d(X)\}}{\pi_d(X)} \int_0^\infty \frac{\mathrm{d}M_C(u\,|\,A,X)}{S_C(u\,|\,A,X)} Q(u, A, X).
\end{split}
\end{align}
where $\mu(a, x) = \mathbb{E}[y(T) \,|\, A = a, X = x, I_S =1]$ and $Q(u, a, x) = \mathbb{E}[y(T) \,|\, T \geq u, A = a, X = x, I_S = 1]$ \footnote{Note that $\mathbb{E}[y(T) \,|\, T \geq u, A, X] = -\int_u^\infty y(s) \, \mathrm{d}S(s \,|\, A, X) / S(u\,|\, A, X)$. For instance, when $y(T) = I\{T \geq t\}$, we have $\mathbb{E}[y(T) \,|\, T \geq u, A, X] = S(t\,|\, A, X)/S(u \,|\, A, X)$ for $u \leq t$.}.
\end{proposition}

The semiparametric EIF guides us in constructing efficient estimators combining the source and target samples. Compared to \eqref{eq:dr.orig}, this EIF captures additional covariates information from the target population via the outcome model and thus removes the sampling bias. An efficient estimation procedure is proposed in the next section, and we show that it enjoys the double robustness property, i.e., it is consistent if either the survival outcome models $\mu(a, x), Q(u, a, x)$ or the models of propensity score $\pi_A(x)$, sampling score $\pi_S(x)$ and censoring process $S_C(t \,|\, a, x)$ are correct. Moreover, this EIF is Neyman orthogonal in the sense discussed in \citet{chernozhukov2018double}. Therefore, a cross-fitting procedure is also proposed, allowing flexible machine learning methods for the nuisance parameters estimation, and $\sqrt{N}$ rate of convergence can be achieved.

\subsection{An efficient and robust estimation procedure}
\label{subsec:erp}

In this section, we focus on estimating the survival function $S_d(t) = Pr(T(d) > t)$ as the value function under ITR $d$. Following the asymptotic linear characterization of survival estimands in \cite{yang2021smim}, our results are readily extended to a broad class of functionals of survival distributions. For instance, the value function of the RMST under ITR $d$ is simply $\int_0^L S_d(t) \mathrm{d}t$. 

Based on the EIF~\eqref{eq:eif}, we propose an estimator for the survival function 
\begin{align}\label{eq:acw}
\begin{split}
\hat{S}_d (t) = & \frac{1}{N} \sum_{i=1}^{N} \bigg\{ \frac{I_{S,i}}{\hat{\pi}_S(X_i)} \frac{I\{A_i = d(X_i)\}}{\hat{\pi}_{d}(X_i)} \frac{\Delta_i \, Y_i(t)}{\hat{S}_C(t \,|\, A_i,X_i)} \\
& \qquad\quad + \left(I_{T,i} \, e(X_i) - \frac{I_{S,i}}{\hat{\pi}_S(X_i)} \frac{I\{A_i = d(X_i)\}}{\hat{\pi}_{d}(X_i)}\right) \hat{S}(t \,|\, A = d(X_i), X_i) \\
& \qquad\quad + \frac{I_{S,i}}{\hat{\pi}_S(X_i)} \frac{I\{A_i = d(X_i)\}}{\hat{\pi}_{d}(X_i)} \int_0^\infty \frac{\hat{S}(t\,|\,A_i,X_i) \mathrm{d}\hat{M}_C(u\,|\,A_i,X_i)}{\hat{S}(u\,|\,A_i,X_i) \hat{S}_C(u\,|\,A_i,X_i)} \bigg\},
\end{split}
\end{align}
where $S(t \,|\, a, x) = Pr(T > t \,|\, A = a, X = x, I_S = 1)$ is the treatment-specific conditional survival function. We posit (semi)parametric models for the nuisance parameters. Let $\pi_A(X;\theta)$ be the posited propensity score model, for example, using logistic regression $\rm{logit}\{\pi_A(X;\theta)\} = \theta^T \tilde{X}$, where $\rm{logit}(x) = \log\{x/(1 - x)\}$. We use the Cox proportional hazard model $\Lambda(t \,|\, A = a, X = x) = \Lambda_{0,a}(t) \exp(\beta_a^T x)$ to estimate the survival functions $S(t \,|\, a, x) = \exp\{-\Lambda(t \,|\, a, x)\}$ and the cumulative baseline hazard function $\Lambda_{0,a}(t) = \int_0^t \lambda_{0,a}(u) \mathrm{d}u$ can be estimated by the Breslow estimator \citep{breslow1972contribution}. Similarly, we posit a Cox proportional hazard model for the censoring process $\Lambda_C(t \,|\, A = a, X = x) = \Lambda_{C0,a}(t) \exp(\alpha_a^T x)$, and the cumulative baseline hazard function $\Lambda_{C0,a}(t)$ is estimated by the Breslow estimator. The sampling score estimation is discussed in the next section.

Let $\hat{S} (t; \eta) = \hat{S}_{d_\eta} (t)$ be the estimated value function for the ITR class $\mathcal{D}_\eta$, then the optimal ITR is given by $d_{\hat{\eta}}(x)$, where $\hat{\eta} = \arg\max_\eta \hat{S} (t; \eta)$.

\subsection{Calibration weighting}

To correct the bias due to the covariate shift between populations, most existing methods directly model the sampling score \citep{cole2010generalizing}, i.e., inverse probability of sampling weighting (IPSW). However, the IPSW method requires the sampling score model to be correctly specified, and it could also be numerically unstable. Alternatively, we introduce the calibration weighting (CW) approach motivated by the identity \eqref{eq:cw.iden}, which is similar to the entropy balancing method \citep{hainmueller2012entropy}. 

Let $\mathbf{g}(X)$ be a vector of functions of $X$ to be calibrated, such as the moments, interactions, and non-linear transformations of $X$. Each subject $i$ in the source sample is assigned a weight $q_i$ by solving the following optimization task:
\begin{align}
\min_{q1,\ldots,q_n} & \sum_{i=1}^n q_i \log q_i, \label{eq:eb} \\
\text{subject to } & q_i \geq 0, \, \sum_{i=1}^n q_i = 1, \, \sum_{i=1}^n q_i \mathbf{g}(X_i) = \tilde{\mathbf{g}}, \label{eq:cw.cst}
\end{align}
where $\tilde{\mathbf{g}} = \sum_{i=n+1}^{n+m} e(X_i) \mathbf{g}(X_i) / \sum_{i=n+1}^{n+m} e(X_i)$ is a design-weighted estimate of $\mathbb{E}[\mathbf{g}(X)]$. The objective function~\eqref{eq:eb} is the negative entropy of the calibration weights, which ensures that the empirical distribution of the weights is not too far away from the uniform, such that it minimizes the variability due to heterogeneous weights. The final balancing constraint in \eqref{eq:cw.cst} calibrates the covariate distribution of the weighted source sample to the target population in terms of $\mathbf{g}(X)$. By introducing the Lagrange multiplier $\lambda$, the minimizer of the optimization task is $q_i = \exp\{\hat{\lambda}^T \mathbf{g}(X_i)\} / \sum_{i=1}^n \exp\{\hat{\lambda}^T \mathbf{g}(X_i)\}$, where $\hat{\lambda}$ solves the estimating equation $\sum_{i=1}^n \exp\{\lambda^T \mathbf{g}(X_i)\} \{\mathbf{g}(X_i) - \tilde{\mathbf{g}}\} = 0$. Since we only require specifying $\mathbf{g}(X)$, calibration weighting avoids explicitly modeling the sampling score and evades extreme weights. 

Moreover, suppose that the sampling score follows a loglinear model $\pi_S (X;\lambda) = \exp\{\lambda^T \tilde{X}\}$, \cite{lee2021improving,lee2022generalizable} show that there is a direct correspondence between the calibration weights and the estimated sampling score, i.e., $q_i = \{N \pi_S(X_i;\hat{\lambda)}\}^{-1} + o_p(N^{-1})$. We also note that if the fraction $n/N$ is small, the loglinear model is close to the widely used logistic regression model; our simulation studies show the robustness of calibration weights.

\begin{remark}
Other objective functions can also be used for calibration weights estimation. \cite{chu2022targeted} considers a generic convex distance function $h(q)$ from the Cressie and Read family of discrepancies \citep{cressie1984multinomial}. Thus the optimization task is $\min_{q1,\ldots,q_n} \sum_{i=1}^n h(q_i)$ under the constraints~\eqref{eq:cw.cst}, and the correspondence between the sampling score model $\pi_S$ and the objective function $h$ has also been established.
\end{remark}

\subsection{Cross-fitting}
\label{subsec:cf}

Utilizing the Neyman orthogonality of EIF \eqref{eq:eif}, we consider flexible machine learning methods for estimating the nuisance parameters, where we want to remain agnostic on modeling assumptions for the complex treatment assignment, survival, and censoring processes. There is extensive recent literature on nonparametric methods for heterogeneous treatment effect estimation with survival outcomes. \cite{cui2020estimating} extends the generalized random forests \citep{athey2019generalized} to estimate heterogeneous treatment effects in a survival and observational setting. See \cite{xu2022treatment} for details and practical considerations. A description of the proposed cross-fitting procedure is given below \citep{schick1986asymptotically,chernozhukov2018double}. Throughout, we use the subscript $CF$ to denote the cross-fitted version. 
\begin{algorithm}\label{acw.crossfit}
\caption{Pseudo algorithm for the cross-fitting procedure}
\medskip
\begin{description}
\item[Step 1] Randomly split the datasets $\mathcal{O}_s$ and $\mathcal{O}_t$ respectively into $K$-folds with equal size such that $\mathcal{O}_s = \cup_{k=1}^K \mathcal{O}_{s,k}, \mathcal{O}_t = \cup_{k=1}^K \mathcal{O}_{t,k}$. For each $k \in \{1,\ldots,K\}$, let $\mathcal{O}_{s,k}^c = \mathcal{O}_s \backslash \mathcal{O}_{s,k}, \mathcal{O}_{t,k}^c = \mathcal{O}_s \backslash \mathcal{O}_{t,k}$.
\item[Step 2] For each $k \in \{1,\ldots,K\}$, estimate the nuisance parameters only using data $\mathcal{O}_{s,k}^c$ and $\mathcal{O}_{t,k}^c$; then obtain an estimate of the value function $\hat{V}_{CF,k} (\eta)$ using data $\mathcal{O}_{s,k}$.
\item[Step 3] Aggregate the estimates from $K$ folds: $\hat{V}_{CF}(\eta) = \frac{1}{K} \sum_{k=1}^K \hat{V}_{CF,k} (\eta)$.
\item[Step 4] The estimated optimal ITR is indexed by $\hat{\eta} = \arg\max_\eta \hat{V}_{CF} (\eta)$.
\end{description}
\end{algorithm}

\section{Asymptotic properties}
\label{sec:asym}

In this section, we present the asymptotic properties of the proposed methods. To establish the asymptotic properties, we require the following assumptions.
\begin{assumption}\label{asmp:regu}
(i) The value function $V(\eta)$ is twice continuously differentiable in a neighborhood of $\eta^\ast$. (ii) There exists some constant $\delta_0 > 0$ such that $Pr(0 < |\tilde{X}^T \eta| < \delta) = O(\delta)$, where the big-$O$ term is uniform in $0 < \delta < \delta_0$.
\end{assumption}

Condition \textit{(i)} is a standard regularity condition to establish uniform convergence. Similar margin conditions as \textit{(ii)}, which state that $Pr(0 < |\gamma (X)| < \delta) = O(\delta^\alpha)$ \footnote{Let $\gamma(X) = \mathbb{E}[T \,|\, A = 1, X] - \mathbb{E}[T \,|\, A = 0, X]$ denote the conditional average treatment effect, then the optimal ITR in an unrestricted class is given by $d(X) = I\{\gamma (X) > 0\}$.}, are often assumed in the literature of classification \citep{tsybakov2004optimal,audibert2007fast}, reinforcement
learning \citep{farahmand2011action,hu2021fast} and optimal treatment regimes \citep{luedtke2016statistical,luedtke2020performance}, to guarantee a fast convergence rate. Note that $\alpha = 0$ imposes no restriction, which allows $\gamma(X) = 0$ almost surely, i.e., the challenging setting of exceptional laws where the optimal ITR is not uniquely defined \citep{robins2004optimal,robins2014discussion}, while the case $\alpha = 1$ is of particular interest and would hold if $\gamma(X)$ is absolutely continuous with bounded density.

\begin{theorem}\label{thm:para}
Under Assumptions \ref{asmp:cnpc} - \ref{asmp:regu} and standard regularity conditions provided in the Supplementary Material, if either the survival outcome model, or the models of the propensity score, the sampling score and the censoring process are correct, we have that as $N \rightarrow \infty$, (i) $\hat{S}(t;\eta) \to S(t;\eta)$ for any $\eta$ and $0 < t \leq L$; (ii) $\sqrt{N} \left\{\hat{S}(t;\eta) - S(t;\eta)\right\}$ converges weakly to a mean zero Gaussian process for any $\eta$; (iii) $N^{1/3} \left\|\hat{\eta} - \eta^\ast\right\|_2 = O_p(1)$; (iv) $\sqrt{N} \left\{\hat{S}(t; \hat{\eta}) - S(t; \eta^\ast)\right\} \to \mathcal{N}(0, \sigma_{t,1}^2)$, where $\sigma_{t,1}$ is given in the Supplementary Material.
\end{theorem}

Next, to characterize the asymptotic behavior of the estimator with the nonparametric estimation of nuisance parameters, we assume the following consistency and convergence rate conditions of the nonparametric plug-in nuisance estimators.
\begin{assumption}\label{asmp:nonp}
Assume the following convergences in probability: $\sup_{x \in \mathcal{X}} |\hat{\pi}_A(x) - \pi_A(x)| \to 0$, $\sup_{x \in \mathcal{X}} |\hat{\pi}_S(x) - \pi_S(x)| \to 0$, and for $a = 0, 1$,
\begin{gather*}
\sup_{x \in \mathcal{X}, u \leq h} |\hat{S}_C(u \,|\, a, x) - S_C(u \,|\, a, x)| \to 0, \sup_{x \in \mathcal{X}, u \leq h} \left|\frac{\hat{\lambda}_C(u \,|\, a, x)}{\hat{S}_C(u \,|\, a, x)} - \frac{\lambda_C(u \,|\, a, x)}{S_C(u \,|\, a, x)}\right| \to 0, \\
\sup_{x \in \mathcal{X}} |\hat{\mu}(a, x) - \mu(a, x)| \to 0, \sup_{x \in \mathcal{X}, u \leq h} |\hat{Q}(u, a, x) - Q(u, a, x)| \to 0;
\end{gather*}
and the following rates of convergence: $\mathbb{E}\left[\sup_{x \in \mathcal{X}} |\hat{\pi}_A(x) - \pi_A(x)|\right] = o_p(n^{-1/4})$, \\
$\mathbb{E}\left[\sup_{x \in \mathcal{X}} |\hat{\pi}_S(x) - \pi_S(x)|\right] = o_p(n^{-1/4})$, and for $a = 0, 1$,
\begin{gather*}
\sup_{u \leq h} \mathbb{E}\left[\sup_{x \in \mathcal{X}} \left|\hat{S}_C(u \,|\, a, x) - S_C(u \,|\, a, x)\right|\right] = o_p(n^{-1/4}), \\
\sup_{u \leq h} \mathbb{E}\left[\sup_{x \in \mathcal{X}} \left|\frac{\hat{\lambda}_C(u \,|\, a, x)}{\hat{S}_C(u \,|\, a, x)} - \frac{\lambda_C(u \,|\, a, x)}{S_C(u \,|\, a, x)}\right|\right] = o_p(n^{-1/4}), \\
\mathbb{E}\left[\sup_{x \in \mathcal{X}} |\hat{\mu}(a, x) - \mu(a, x)|\right] = o(n^{-1/4}), \sup_{u \leq h} \mathbb{E}\left[\sup_{x \in \mathcal{X}} |\hat{Q}(u, a, x) - Q(u, a, x)|\right] = o(n^{-1/4}).
\end{gather*}

\end{assumption}

The rate conditions in Assumption~\ref{asmp:nonp} are generally assumed in the literature \citep{kennedy2022semiparametric}. This rate can be achieved by many existing methods under certain structural assumptions on the nuisance parameters. Note that the nuisance parameters do not necessarily need to be estimated at the same rates $n^{-1/4}$ for our theorems to hold; it would suffice that the product of rates of any combination of two nuisance parameters is $n^{-1/2}$. 

\begin{theorem}\label{thm:np}
Under Assumptions \ref{asmp:cnpc} - \ref{asmp:nonp}, we have that as $N \to \infty$, (i) $\hat{S}_{CF}(t;\eta) \to S(t;\eta)$ for any $\eta$ and $0 < t \leq L$; (ii) $\sqrt{N} \left\{\hat{S}_{CF}(t;\eta) - S(t;\eta)\right\}$ converges weakly to a mean zero Gaussian process for any $\eta$; (iii) $N^{1/3} \|\hat{\eta} - \eta^\ast\|_2 = O_p(1)$; (iv) $\sqrt{N} \left\{\hat{S}_{CF}(t;\hat{\eta}) - S(t;\eta^\ast)\right\} \to \mathcal{N}(0, \sigma_{t,2}^2)$, where $\sigma_{t,2}$ is given in the Supplementary Material.
\end{theorem}

Besides the survival functions, another common measure of particular interest in survival analysis is the RMST. Let $V_{\text{RMST}} (\eta) = \mathbb{E}[\min(T(d_\eta), L)]$. We present two corollaries.

\begin{corollary}\label{cor:para.rmst}
Under Assumptions \ref{asmp:cnpc} - \ref{asmp:regu} and standard regularity conditions provided in the Supplementary material, if either the survival outcome model or the models of the propensity score, the censoring and sampling processes are correct, we have that as $N \to \infty$, (i) $\hat{V}_{\text{RMST}}(\eta) \to V_{\text{RMST}}(\eta)$ for any $\eta$; (ii) $N^{1/3} \|\hat{\eta} - \eta^\ast\|_2 = O_p(1)$; (iii) $\sqrt{N} \left\{\hat{V}_{\text{RMST}}(\hat{\eta}) - V_{\text{RMST}}(\eta^\ast)\right\} \to \mathcal{N}(0, \sigma_{3}^2)$, where $\sigma_{3}$ is given in the Supplementary Material.
\end{corollary}

\begin{corollary}\label{cor:np.rmst}
Under Assumptions \ref{asmp:cnpc} - \ref{asmp:nonp}, we have that as $N \to \infty$, (i) $\hat{V}_{\text{RMST},CF}(\eta) \to V_{\text{RMST}}(\eta)$ for any $\eta$; (ii) $N^{1/3} \|\hat{\eta} - \eta^\ast\|_2 = O_p(1)$; (iii) $\sqrt{N} \left\{\hat{V}_{\text{RMST},CF} (\hat{\eta}) - V_{\text{RMST}} (\eta^\ast)\right\} \to \mathcal{N}(0, \sigma_{4}^2)$, where $\sigma_{4}$ is given in the Supplementary Material..
\end{corollary}

Finally, we show that when the covariate distributions of the source and target populations are the same, the semiparametric efficiency bounds of $\hat{V}_{DR} (\eta)$ and $\hat{V}_{CF} (\eta)$ are equal.

\begin{theorem}\label{thm:st}
Under Assumptions \ref{asmp:cnpc} - \ref{asmp:nonp}, when the covariate distributions of the source and target populations are the same, both $\sqrt{N} \{\hat{V}_{DR} (\eta) - V (\eta)\}$ and $\sqrt{N} \{\hat{V}_{CF} (\eta) - V (\eta)\}$ are asymptotically normal with mean zero and same variance.
\end{theorem}

Theorem~\ref{thm:st} implies that when there is no covariate shift, our proposed estimator does not lose efficiency in comparison to the original double robust estimator since the augmentation term in EIF~\eqref{eq:eif} from the target population, $I_T \, e(X) \mu(d(X), X)$, is asymptotically equal to this term evaluated on the source population in this case.

Moreover, when the covariate shift exists, we consider the optimal ITR $d^{\text{opt}}$ without restriction on the ITR class.

\begin{theorem}\label{thm:tl}
Under Assumptions \ref{asmp:cnpc} - \ref{asmp:nonp}, If $d^{\text{opt}} \in \mathcal{D}_\eta$, i.e., $d^{\text{opt}} = d_{\eta^{\ast}}$, both the maximizers of $\hat{V}_{DR} (\eta)$ and $\hat{V}_{CF} (\eta)$ converge to $\eta^{\ast}$. However, $\hat{V}_{DR} (\eta)$ is a biased estimator of $V(\eta)$.
\end{theorem}

Theorem~\ref{thm:tl} implies if the true optimal ITR belongs to the restricted ITR class $\mathcal{D}_\eta$, standard methods, without accounting for the covariate shift, are still able to recover the optimal ITR but fail to be consistent for the value function, due to the covariate shift. And we can only rely on the proposed method to draw valid inferences.

\section{Simulation}
\label{sec:simu}

In this section, we investigate the finite-sample properties of our method through extensive numerical simulations \footnote{The \texttt{R} code to replicate all results is available at \url{https://github.com/panzhaooo/transfer-learning-survival-ITR}.}.

Consider a target population of sample size $N = 2 \times 10^5$. The covariates $(X_1, X_2, X_3)^T$ are generated from a multivariate normal distribution with mean $0$, unit variance with $corr(X_1, X_3) = 0.2$ and all other pairwise correlations equal to $0$, and further truncated below $-4$ and above $4$ to satisfy regularity conditions. The target sample is a random sample of size $m = 8000$ from the target population. The sampling score follows $\pi_S(X) = \text{expit}(-4.5 - 0.5 X_1 - 0.5 X_2 - 0.4 X_3)$; thus the source sampling rate is around $1.6\%$, and the source sample size around $n = 3000$. The treatment assignment mechanism in the source sample follows $\pi_A(X) = \text{expit}(0.5 + 0.8 X_1 - 0.5 X_2)$.

The counterfactual survival times $T(a)$ are generated according to the hazard functions $\lambda(t \,|\, A = 0, X) = \exp(t) \cdot \exp(-2.5 -1.5 X_1 - X_2 - 0.7 X_3)$ and $\lambda(t \,|\, A = 1, X) = \exp(t) \cdot \exp(-1 - X_1 - 0.9 X_2 - X_3 - 2 X_2^2 + X_1 X_3)$. The censoring time $C$ is generated according to the hazard functions $\lambda_C(t \,|\, A = 0, X) = 0.04 \exp(t) \cdot \exp(-1.6 + 0.8 X_1 - 1.1 X_2 - 0.7 X_3)$ and $\lambda_C(t \,|\, A = 1, X) = 0.04 \exp(t) \cdot \exp(-1.8 - 0.8 X_1 - 1.7 X_2 - 1.4 X_3)$. The resultant censoring rate is approximately $20\%$.

We consider the RMST with the maximal time horizon $L = 4$ as the value function. To evaluate the performance of different estimators for optimal ITRs, we compute the corresponding true value functions and percentages of correct decisions (PCD) for the target population. Specifically, we generate a large sample with size $\tilde{N} = 1 \times 10^5$ from the target population. The true value function of any ITR $d(\cdot \,;\, \eta)$ is computed by $V(\eta) = \tilde{N}^{-1} \sum_{i=1}^{\tilde{N}} \min\{d(X_i \,;\, \eta) T_i(1) + (1 - d(X_i \,;\, \eta)) T_i(0), L\}$ and its associated PCD is computed by $1 - \tilde{N}^{-1} \sum_{i=1}^{\tilde{N}} | d(X_i \,;\, \eta^\ast) - d(X_i \,;\, \eta)|$, where $\eta^\ast = \arg\max_\eta V(\eta)$.

We compare the following estimators for the RMST $\hat{V} (\eta) = \int_0^L \hat{S} (t ; \eta) \mathrm{d}t$:
\begin{itemize}
    \item Naive: $\hat{S}^{\text{Naive}} (t ; \eta) = \frac{1}{n} \sum_{i=1}^n \frac{I\{A_i = d(X_i)\}}{\hat{\pi}_d(X_i)} \frac{\Delta_i Y_i (t)}{\hat{S}_C (U \,|\, A, X)}$; IPW formula~\eqref{eq:ipw} without using the sampling score;
    \item IPSW: $\hat{S}^{\text{IPSW}} (t ; \eta) = \frac{1}{n} \sum_{i=1}^n \frac{I_{S,i}}{\hat{\pi}_S (X_i)} \frac{I\{A_i = d(X_i)\}}{\hat{\pi}_d(X_i)} \frac{\Delta_i Y_i (t)}{\hat{S}_C (U \,|\, A, X)}$; IPW formula~\eqref{eq:ipw} where the sampling score is estimated via logistic regression; 
    \item CW-IPW: $\hat{S}^{\text{CW-IPW}} (t ; \eta) = \sum_{i=1}^n q_i \, \frac{I\{A_i = d(X_i)\}}{\hat{\pi}_d(X_i)} \frac{\Delta_i Y_i (t)}{\hat{S}_C (U \,|\, A, X)}$ IPW formula~\eqref{eq:ipw} where the sampling score is estimated by calibration weighting;
    \item CW-OR: $\hat{S}^{\text{CW-OR}} (t ; \eta) = \sum_{i=1}^n q_i \, \hat{S} (t \,|\, A = d(X_i), X_i)$; OR formula~\eqref{eq:or} in combination with calibration weights by the identity~\eqref{eq:cw.iden};
    \item ORt: $\hat{S}^{\text{ORt}} (t ; \eta) = \frac{1}{m} \sum_{i=n+1}^{n+m} \hat{S} (t \,|\, A = d(X_i), X_i)$; OR formula~\eqref{eq:or} evaluated on the target sample;
    \item ACW: augmented estimator~\eqref{eq:acw}, where the sampling score is estimated by calibration weighting.
\end{itemize}

\begin{remark}
Since the estimated value functions are non-convex and non-smooth, multiple local optimal may exist in the optimization task, and many derivatives-based algorithms do not work for this challenging setting. Here we utilize the genetic algorithm implemented in the R package \texttt{rgenoud} \citep{mebane2011genetic}, which performs well in our numerical experiments. We refer to \citet{mitchell1998introduction} for algorithmic details.  
\end{remark}

\subsection{(Semi)parametric models}

We first consider the setting where the nuisance parameters are estimated by posited (semi)parametric working models as introduced in Section~\ref{subsec:erp}. To assess the performance of these estimators under model misspecification, we consider four scenarios: (1) all models are correct, (2) only the survival outcome model is correct, (3) only the survival outcome model is wrong, (4) all models are wrong. For the wrong sampling model, the weights are estimated using calibration on $e^{X_1}$. The wrong propensity score model is fitted on $e^{X_3}$. The wrong Cox models for survival and censoring times are fitted on $(e^{X_1}, e^{X_2}, e^{X_3})^T$.

Figure~\ref{fig:simu.para} and Table~\ref{tab:sim.para} report the simulation results from $350$ Monte Carlo replications. Variance is estimated by a bootstrap procedure with $B = 200$ bootstrap replicates. The proposed ACW estimator is unbiased in scenarios (1) - (3), and the $95\%$ coverage probabilities approximately achieve the nominal level, which shows the double robustness property. 

\begin{figure}[hp]
    \centering
    \caption{Boxplot of the estimated value, true value and PCD results of estimators under four model specification scenarios. O: survival outcome, S: sampling score, A: propensity score, C: censoring; T: True (correctly specified) model, W: Wrong (misspecified) model.}\bigskip
    \includegraphics{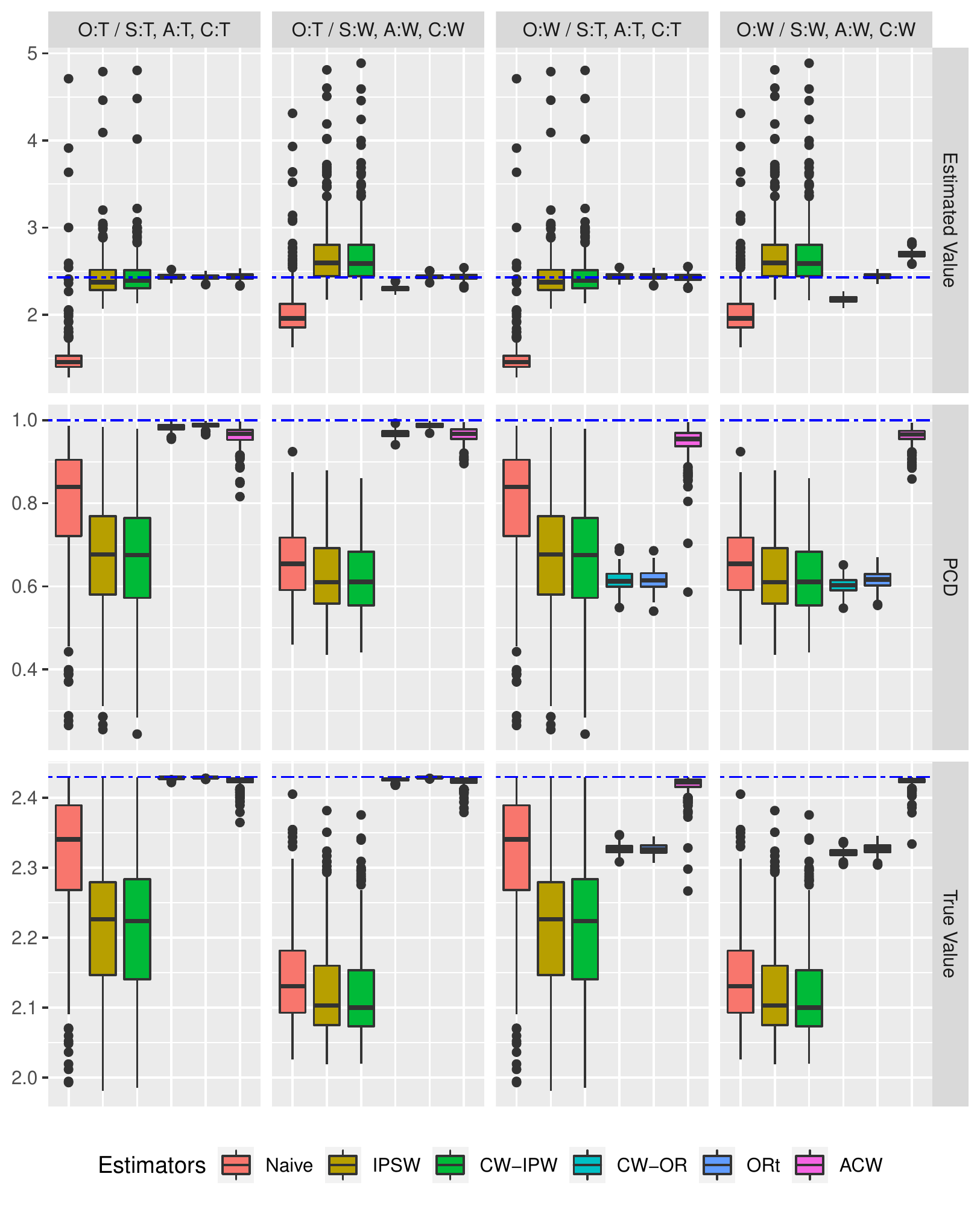}
    \label{fig:simu.para}
\end{figure}

\begin{table}[p]
\centering
\caption{Numerical results under four different model specification scenarios. Bias is the empirical bias of point estimates; SD is the empirical standard deviation of point estimates; SE is the average of bootstrap standard error estimates; CP is the empirical coverage probability of the $95\%$ confidence intervals.}\bigskip
\begin{tabular}{l cccc c cccc}
\hline\hline
 & Bias & SD & SE & CP(\%) && Bias & SD & SE & CP(\%) \\
\hline
 & \multicolumn{4}{c}{O:T / S:T, A:T, C:T} && \multicolumn{4}{c}{O:T / S:W, A:W, C:W} \\
\cline{2-5}\cline{7-10}
Naive & $-0.8801$ & $0.4595$ & $0.2189$ & $7.43$ && $-0.3528$ & $0.5024$ & $0.4598$ & $37.43$ \\
IPSW & $0.0185$ & $0.3685$ & $0.2562$ & $87.14$ && $0.3377$ & $0.7144$ & $0.6958$ & $98.29$ \\
CW-IPW & $0.0378$ & $0.3701$ & $0.2498$ & $88.29$ && $0.3406$ & $0.7144$ & $0.6957$ & $97.71$ \\
CW-OR & $0.0047$ & $0.0273$ & $0.0286$ & $96.29$ && $-0.1312$ & $0.0269$ & $0.0279$ & $0.57$ \\
ORt & $0.0041$ & $0.0258$ & $0.0262$ & $95.14$ && $0.0035$ & $0.0258$ & $0.0262$ & $95.71$ \\
ACW & $0.0070$ & $0.0380$ & $0.0369$ & $94.29$ && $0.0055$ & $0.0316$ & $0.0334$ & $95.43$ \\
\hline
 & \multicolumn{4}{c}{O:W / S:T, A:T, C:T} && \multicolumn{4}{c}{O:W / S:W, A:W, C:W} \\
\cline{2-5}\cline{7-10}
Naive & $-0.8801$ & $0.4595$ & $0.2207$ & $6.86$ && $-0.3528$ & $0.5024$ & $0.5018$ & $38.57$ \\
IPSW & $0.0185$ & $0.3685$ & $0.2486$ & $87.71$ && $0.3377$ & $0.7144$ & $0.7586$ & $99.14$ \\
CW-IPW & $0.0378$ & $0.3701$ & $0.2418$ & $88.86$ && $0.3406$ & $0.7144$ & $0.7570$ & $98.57$ \\
CW-OR & $0.0103$ & $0.0370$ & $0.0362$ & $92.29$ && $-0.2551$ & $0.0366$ & $0.0391$ & $0.00$ \\
ORt & $0.0094$ & $0.0365$ & $0.0355$ & $94.00$ && $0.0115$ & $0.0328$ & $0.0355$ & $95.71$ \\
ACW & $-0.0010$ & $0.0426$ & $0.0419$ & $93.14$ && $0.2644$ & $0.0422$ & $0.0475$ & $0.57$ \\
\hline\hline
\end{tabular}
\label{tab:sim.para}
\end{table}

\subsection{Flexible machine learning methods}

When utilizing flexible ML methods, we construct the cross-fitted ACW estimator as introduced in Section~\ref{subsec:cf}. The data generation process is the same as above, except that the censoring time $C$ is generated according to the hazard functions $\lambda_C(t \,|\, A = 0, X) = 0.2 \exp(t) \cdot \exp(-1.6 + 0.8 X_1 - 1.1 X_2 - 0.7 X_3)$ and $\lambda_C(t \,|\, A = 1, X) = 0.2 \exp(t) \cdot \exp(-1.8 - 0.8 X_1 - 1.7 X_2 - 1.4 X_3)$ which leads to an increased censoring rate of approximately $33\%$, so there are enough observations to get an accurate estimate of the censoring process. The propensity score is estimated by the generalized random forest. The conditional survival and censoring functions are estimated by the random survival forest. The calibration weighting uses calibration on the first- and second-order moments of $X$.

First, we study the impact of sample sizes on the performance of the ML methods, and simulation results are given in the Supplementary Material. With a small sample size, the ACW estimator is largely biased, and the bias diminishes as the sample size increases.

Next, we compare the performance of different estimators with target population size $N = 6 \times 10^5$ and target sample size $m = 24000$. Figure~\ref{fig:sim.np2} shows the simulation results from $200$ Monte Carlo replications. The two IPW-based estimators are biased and perform poorly due to the large variability of weights. The two OR-based estimators have comparable performance as the ACW estimator in terms of PCD and true value function but still suffer from the overfitting bias. Only the ACW estimator is consistent and provides valid inferences.

\begin{figure}[ht]
    \centering
    \caption{Boxplots of the estimated value, true value, and PCD of different estimators using flexible ML methods. }
    \includegraphics{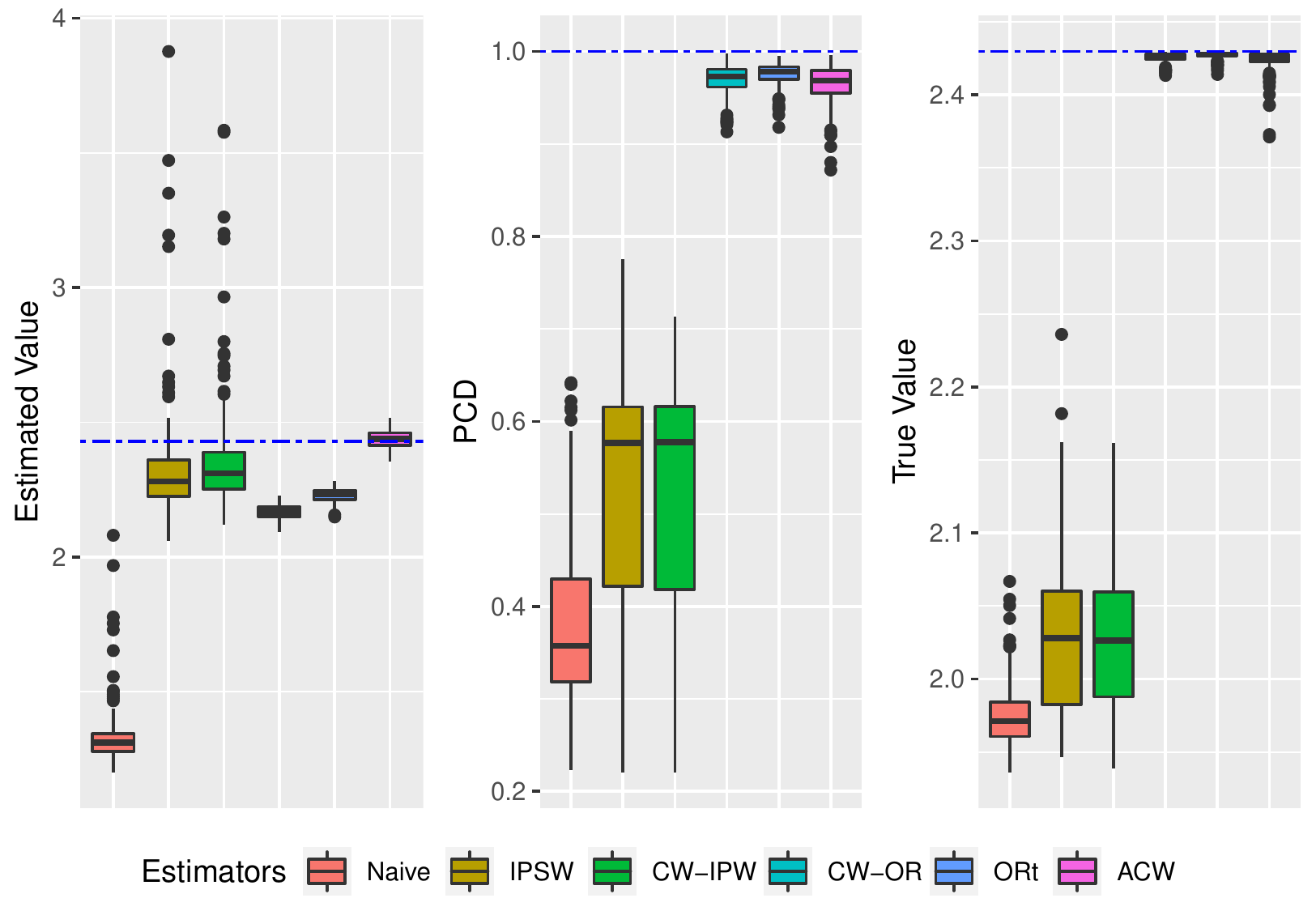}
    \label{fig:sim.np2}
\end{figure}

\section{Real Data Analysis}
\label{sec:rdat}

In this section, to illustrate the proposed method, we study the sodium bicarbonate therapy for patients with severe metabolic acidaemia in the intensive care unit by leveraging the RCT data BICAR-ICU \citep{jaber2018sodium} and the observational study (OS) data from \cite{jung2011severe}. Specifically, we consider the BICAR-ICU data as the source sample and the observational study data as the target sample. The BICAR-ICU is a multi-center, open-label, randomized controlled, phase 3 trial between May 5, 2015, and May 7, 2017, which includes $387$ adult patients admitted within $48$ hours to the ICU with severe acidaemia. The prospective, multiple-center observational study was conducted over thirteen months in five ICUs, consisting of $193$ consecutive patients who presented with severe acidemia within the first $24$ hours of their ICU admission. Some heterogeneity exists between the two populations.

Both the RCT and OS datasets contain detailed measurements of ICU patients with severe acidaemia. Motivated by the clinical practice and existing work in the medical literature, we consider ITRs that depend on the following five variables: SEPSIS, AKIN, SOFA, SEX, and AGE. A detailed description of the data preprocessing and variable selection is given in the Supplementary Material. Table~\ref{tab:data} summarizes the baseline characteristics of the two datasets. The baseline covariates distribution of the patients in the BICAR-ICU differs from the distribution in the observational study; specifically, the BICAR-ICU patients have higher SOFA scores and the more frequent presence of acute kidney injury and sepsis.

\begin{table}[ht]
\centering
\caption{Summary of baseline characteristics of the BICAR-ICU trial sample and the OS sample. Mean (standard deviation) for continuous and number (proportion) for the binary covariate. }\bigskip
\resizebox{\textwidth}{!}{
\begin{tabular}{l ccccc}
\hline\hline
 & SEPSIS & AKIN & SOFA & SEX & AGE \\
\hline
BICAR-ICU $(n = 387)$ & $236 \, (60.98\%)$ & $181 \, (46.77\%)$ & $10.12 \, (3.72)$ & $237 \, (61.24\%)$ & $63.95 \, (14.41)$ \\
OS $(m = 193)$ & $99 (51.30\%)$ & $75 \, (38.86\%)$ & $9.10 \, (4.54)$ & $122 \, (63.21\%)$ & $62.73 \, (17.49)$ \\
\hline\hline
\end{tabular}}
\label{tab:data}
\end{table}

We apply our proposed ACW estimator to learn the optimal ITR for the target population. The calibration weights are estimated based on the means of continuous covariates and the proportions of the binary covariates. The propensity score is estimated using a logistic regression model, and the Cox proportional hazard model is fitted for the survival outcome with all covariates. The censoring only occurred on the $28$th day when the follow-up in ICU ends. We consider the class of linear ITRs that depend on all five variables:
\begin{equation*}
\mathcal{D} = \{I\{\eta_1 + \eta_2 \text{SEPSIS} + \eta_3 \text{AKIN} + \eta_4 \text{SOFA} + \eta_5 \text{SEX} + \eta_6 \text{AGE} > 0 \} : \eta_1, \ldots, \eta_6 \in \mathbb{R}, |\eta_6| = 1 \},
\end{equation*}
with the aim to maximize the RMST within $28$ days in ICU stay. The estimated parameter indexing the optimal ITR is $\hat{\eta}_{\text{ACW}} = (22.9, -36.1, 87.4, -9.8, 33.7, 1.0)^T$, which leads to an estimated value function $\hat{V}(\hat{\eta}_{\text{ACW}}) = 19.52$ days, with confidence interval $[17.74, 21.30]$ given by $200$ bootstraps. In contrast, we also use the standard double robust method to estimate the optimal ITR for the RCT, indexed by $\hat{\eta}_{\text{DR.RCT}}$ which maximize the value function $\hat{V}_{\text{DR}} (\eta)$ in \eqref{eq:dr.orig} with $y(T) = \min(T, 28)$. The estimated value function is $\hat{V}(\hat{\eta}_{\text{DR.RCT}}) = 15.37$ days for the target population.

\section{Discussion}
\label{sec:disc}

In this paper, we present an efficient and robust transfer learning framework for estimating optimal ITR with right-censored survival data that generalizes well to the target population. The proposed method can be improved or extended in several directions for future work. Construction and estimation of optimal ITRs for multiple decision points with censored survival data are challenging, taking into account the timing of censoring, events and decision points \citep{jiang2017estimation,hager2018optimal}, e.g., using a reinforcement learning method \citep{cho2020multi}. Furthermore, besides the class of ITRs indexed by a Euclidean parameter, it may be possible to consider other classes of ITRs, such as tree or list-based ITRs. The current work focus on value functions in the form $V(d) = \mathbb{E}[y(T(d))]$ and can also be modified in case of optimizing certain easy-to-interpret quantile criteria, which does not require specifying an outcome regression model and is robust for heavy-tailed distributions \citep{zhou2022transformation}. And relaxing the restrictive assumptions such as positivity \citep{yang2018trimming,jin2022policy} and unconfoundedness \citep{cui2021semiparametric,qi2021proximal} for learning optimal ITRs is also a fruitful direction.

\subsubsection*{Acknowledgments}

Josse and Zhao gratefully acknowledge the French National Research Agency ANR-16-IDEX-0006.
Yang is partially supported by the USA National Institutes of Health NIA grant 1R01AG066883 and NIEHS grant 1R01ES031651.

The authors thank Maxime Fosset and Boris Jung for their help and support interpreting the BICAR-ICU trial and observational study data.

\bibliographystyle{agsm}
\bibliography{Bibliography}

\appendix\newpage

\bigskip
\begin{center}
{\large\bf SUPPLEMENTARY MATERIAL}
\end{center}

\section{Preliminaries}

\subsection{Counting processes for Cox model}
\label{subsec:cp.cox}

We use the counting process theory of \cite{andersen1982cox} in our theoretical framework to study the large sample properties of Cox model. We state the existing results that are used in our proof.

Let $X^{\otimes l}$ denote $1$ for $l=0$, $X$ for $l=1$, and $X X^T$ for $l=2$. Define
\begin{equation*}
U_a^{(l)}(\beta_a, t) = \frac{1}{n_a}\sum_{i=1}^{n} I\{A_i = a\} X_i^{\otimes l} \exp(\beta_a^T X_i) Y_i(t) \text{ and } u_a^{(l)}(\beta_a, t) = \mathbb{E}\left[X^{\otimes l} \exp(\beta_a^T X) Y(t)\right],
\end{equation*}
where $n_a = \sum_{i=1}^{n} I\{A_i = a\}$, and define 
\begin{equation*}
E_a(\beta_a, t) = \frac{U_a^{(1)}(\beta_a, t)}{U_a^{(0)}(\beta_a, t)} \text{ and } e_a(\beta_a, t) = \frac{u_a^{(1)}(\beta_a, t)}{u_a^{(0)}(\beta_a, t)}.
\end{equation*}

The maximum partial likelihood estimator $\hat{\beta}_a$ for the Cox proportional hazards model solves the estimating equation
\begin{equation*}
\mathcal{S}_{a,n}(\beta_a) = \frac{1}{n_a}\sum_{i=1}^{n} I\{A_i = a\} \int\left\{X_i - \frac{U_1^{(1)}(\beta_a, u)}{U_1^{(0)}(\beta_a, u)}\right\}\mathrm{d}N_i(u) = 0,
\end{equation*}
and the cumulative baseline hazard function $\hat{\Lambda}_{0,a}$ is estimted by the Breslow estimator:
\begin{equation*}
\hat{\Lambda}_{0,a}(t) = \int_0^t \frac{\sum_{i=1}^n I\{A_i=a\} \mathrm{d}N_i (u)}{\sum_{i=1}^n I\{A_i=a\} \exp(\hat{\beta}_a^T X_i) Y_i (u)}, a = 0, 1.
\end{equation*}

Under certain regularity conditions \citep[Conditions A -- D]{andersen1982cox}, $\hat{\beta}_a$ and $\hat{\Lambda}_{0,a}$ converge in probability to the limits $\beta_a^\ast$ and $\Lambda_{0,a}^\ast$, respectively; and we have
\begin{equation*}
\sqrt{n_a}(\hat{\beta}_a - \beta_a^\ast) = \Gamma_a^{-1} \frac{1}{\sqrt{n_a}} \sum_{i=1}^n I\{A_i = a\} H_{a,i} + o_p(1),
\end{equation*}
where $\Gamma_a = \mathbb{E}[-\partial \mathcal{S}_{a,n}(\beta_a^\ast)/\partial \beta_a^{\ast T}]$ is the Fisher information matrix of $\beta_a^\ast$, $H_{a,i} = \int I\{A_i = a\} \{X_i - e_a(\beta_a^\ast,u)\} \mathrm{d}M_{a,i}(u)$ and $\mathrm{d}M_{a,i}(u) = \mathrm{d}N_i(u) - \exp(\beta_a^{\ast T} X_i) Y_i(u) \mathrm{d}\Lambda_{0,a}^\ast(u)$. Moreover, let $S^\ast(t \,|\, a, X) = \exp\{-\Lambda_{0,a}^\ast(t) \exp(\beta_a^{\ast T} X)\}$; it is shown that $\sqrt{n_a} \{\hat{S}(t \,|\, a, X_i) - S^\ast(t \,|\, a, X_i)\}$ converges uniformly to a mean-zero Gaussian process for all $X_i$. 

Specifically, we consider the following expansion that we use in our proof of Theorem~\ref{thm:para} and Corollary~\ref{cor:para.rmst},
\begin{align*}
\hat{S}(t \,|\, a, X_i) - S^\ast(t \,|\, a, X_i) = & -S^\ast(t \,|\, a, X_i) \Lambda_{0,a}^\ast(t) \exp(\beta_a^{\ast T} X_i) X_i^T (\hat{\beta}_a - \beta_a^\ast) \\
& -S^\ast(t \,|\, a, X_i) \exp(\beta_a^{\ast T} X_i) (\hat{\Lambda}_{0,a}(t) - \Lambda_{0,a}^\ast(t)),
\end{align*}
and furthermore
\begin{align*}
\hat{\Lambda}_{0,a}(t) - \Lambda_{0,a}^\ast(t) & = \int_0^t \left\{\frac{n_a^{-1}\sum_{i=1}^n I\{A_i=a\} \mathrm{d}N_i (u)}{U_a^{(0)}(\hat{\beta}_a,u)} - \frac{n_a^{-1}\sum_{i=1}^n I\{A_i=a\} \mathrm{d}N_i (u)}{U_a^{(0)}(\beta_a^\ast,u)}\right\} \\
& \quad + \int_0^t \left\{\frac{n_a^{-1}\sum_{i=1}^n I\{A_i=a\} \mathrm{d}N_i (u)}{U_a^{(0)}(\beta_a^\ast,u)} - \mathrm{d}\Lambda_{0,a}^\ast(t)\right\} \\
& = -\left[\int_0^t \frac{U_a^{(1)}(\beta_a^\ast,u)}{\left\{U_a^{(0)}(\beta_a^\ast,u)\right\}^2}\left\{n_a^{-1}\sum_{i=1}^n I\{A_i=a\} \mathrm{d}N_i (u)\right\}\right]^T \left(\hat{\beta}_a - \beta_a^\ast \right) \\
& \quad + \int_0^t \frac{n_a^{-1}\sum_{i=1}^n I\{A_i=a\} \mathrm{d}M_{a,i}(u)}{U_a^{(0)}(\beta_a^\ast,u)} + o_p(1) \\
& = -\left\{\int_0^t e_a(\beta_a^\ast,u) \mathrm{d}\Lambda_{0,a}^\ast(u) \right\}^T \left(\hat{\beta}_a - \beta_a^\ast \right) \\
& \quad + \int_0^t\frac{n_a^{-1}\sum_{i=1}^n I\{A_i=a\} \mathrm{d}M_{a,i}(u)}{U_a^{(0)}(\beta_a^\ast,u)} + o_p(1).
\end{align*}

Combining the above two equations, we obtain
\begin{equation*}
\begin{split}
& \hat{S}(t \,|\, a, X_i) - S^\ast(t \,|\, a, X_i) \\
& = \left[-S^\ast(t \,|\, a, X_i) \Lambda_{0,a}^\ast(t) \exp(\beta_a^{\ast T} X_i) X_i^T - \left\{\int_0^t e_a(\beta_a^\ast,u) \mathrm{d}\Lambda_{0,a}^\ast(u) \right\}^T\right] \left(\hat{\beta}_a - \beta_a^\ast \right) \\
& \quad + \int_0^t\frac{n_a^{-1}\sum_{i=1}^n I\{A_i=a\} \mathrm{d}M_{a,i}(u)}{U_a^{(0)}(\beta_a^\ast,u)} + o_p(1).
\end{split}
\end{equation*}

\subsection{Cross-fitting}
\label{subsec:prel.cf}

To show the high-level idea of cross-fitting, we state the lemma from \cite{kennedy2020sharp}, which is useful in our proof of Theorem~\ref{thm:np} and Corollary~\ref{cor:np.rmst}.

\begin{lemma}\label{lem:cf}
Consider two independent samples $\mathcal{O}_1 = (O_1, \ldots, O_n)$ and $\mathcal{O}_2 = (O_{n+1}, \ldots, O_{\tilde{n}})$, let $\hat{f}(o)$ be a function estimated from $\mathcal{O}_2$ and $\mathbb{P}_n$ the empirical measure over $\mathcal{O}_1$, then we have
\begin{equation*}
    (\mathbb{P}_n - \mathbb{P})(\hat{f} - f) = O_{\mathbb{P}}\left(\frac{\|\hat{f} - f\|}{\sqrt{n}}\right)
\end{equation*}
\end{lemma}
\begin{proof}
First note that by conditioning on $\mathcal{O}_2$ we obtain
\begin{equation*}
    \mathbb{E} \left\{\mathbb{P}_n(\hat{f} - f) \,\big|\, \mathcal{O}_2\right\} = \mathbb{E} (\hat{f} - f \,|\, \mathcal{O}_2) = \mathbb{P}(\hat{f} - f)
\end{equation*}
and the conditional variance is
\begin{equation*}
    var\{(\mathbb{P}_n - \mathbb{P})(\hat{f} - f) \,|\, \mathcal{O}_2\} = var\{\mathbb{P}_n (\hat{f} - f) \,|\, \mathcal{O}_2\} = \frac{1}{n}var(\hat{f} - f \,|\, \mathcal{O}_2) \leq \|\hat{f} - f\|^2 / n
\end{equation*}
therefore by Chebyshev's inequality we have 
\begin{equation*}
    \mathbb{P}\left\{\frac{|(\mathbb{P}_n - \mathbb{P})(\hat{f} - f)|}{\|\hat{f} - f\|^2 / n} \geq t \right\} = \mathbb{E}\left[\mathbb{P}\left\{\frac{|(\mathbb{P}_n - \mathbb{P})(\hat{f} - f)|}{\|\hat{f} - f\|^2 / n} \geq t \,\bigg|\, \mathcal{O}_2\right\}\right] \leq \frac{1}{t^2}
\end{equation*}
thus for any $\epsilon > 0$ we can pick $t = 1 / \sqrt{\epsilon}$ so that the probability above is no more than $\epsilon$, which yields the result.
\end{proof}

\section{Proof of Proposition~\ref{prop:iden}}

We first show the identification by the outcome regression formula.
\begin{align*}
&V(d) = \mathbb{E}[\mathbb{E}[y(T(d)) \,|\, X]] \\
&= \mathbb{E}[d(X)\mathbb{E}[y(T(1)) \,|\, X] + (1 - d(X))\mathbb{E}[y(T(0)) \,|\, X]] \\
&= \mathbb{E}[d(X)\mathbb{E}[y(T(1)) \,|\, X, I_S = 1] + (1 - d(X))\mathbb{E}[y(T(0)) \,|\, X, I_S = 1]] \\
&= \mathbb{E}[d(X)\mathbb{E}[y(T(1)) \,|\, A = 1, X, I_S = 1] & \\
&\qquad\quad + (1 - d(X))\mathbb{E}[y(T(0)) \,|\, A = 0, X, I_S = 1]] \\
&= \mathbb{E}[d(X)\mathbb{E}[y(T) \,|\, A = 1, X, I_S = 1] + (1 - d(X))\mathbb{E}[y(T) \,|\, A = 0, X, I_S = 1]] \\ 
&= \mathbb{E}[\mathbb{E}[y(T) \,|\, A = d(X), X, I_S = 1]] \\
&= \mathbb{E}[I_T \, e(X)\mathbb{E}[y(T) \,|\, A = d(X), X, I_S = 1]].
\end{align*}

Similarly, we show the identification by the IPW formula.
\begin{align*}
&V(d) = \mathbb{E}[\mathbb{E}[y(T) \,|\, A = d(X), X, I_S =1]] \\
&= \mathbb{E}\left[\frac{I_S}{\pi_S (X)} \mathbb{E}[y(T) \,|\, A = d(X), X, I_S = 1] \right] \\
&= \mathbb{E}\left[\frac{I_S}{\pi_S (X)} \frac{I\{A = d(X)\}}{\pi_d(X)} \frac{\Delta\,y(U)}{S_C(U\,|\,A,X)}\right],
\end{align*}
where the last equation follows from the standard IPTW-IPCW formula \citep{van2003unified}.

\section{Proof of Proposition~\ref{prop:eif}}

While \cite{lee2022generalizable} derived the efficient influence function for the treatment specific survival function, here we derive the EIF for the value function $V (d) = \mathbb{E}[I_T \, e(X) \mu(d(X), X)]$.

First consider the full data $Z = (X, A, T, I_S, I_T)$, and we have the factorization as
\begin{equation*}
p(Z) = \{p(X) \pi_S(X) p(A | X, I_S = 1) p(T | A, X, I_S = 1)\}^{I_S} \{p(X)\}^{I_T}.
\end{equation*}

Since $I_S I_T = 0$, the score function is $S(Z) = S(X, A, T, I_S) + I_T S(X)$. Let $V_\epsilon (d) = \mathbb{E}_\epsilon [I_T \, e(X) \mu_{\epsilon}(d(X), X)]$ denote the parameter of interest evaluated under the law $p_\epsilon(Z)$, where $\epsilon$ indexes a regular parametric submodel such that $p_0(Z)$ is the true data generating law. To establish that $V(d)$ is pathwise differentiable with EIF $\phi_d^F$, we need to show that
\begin{equation*}
\frac{\partial}{\partial\epsilon} V_\epsilon (d) \bigg|_{\epsilon=0} = \mathbb{E}[\phi_d^F S(Z)].
\end{equation*}

First, we compute
\begin{equation*}
\frac{\partial}{\partial\epsilon} V_\epsilon (d) \bigg|_{\epsilon=0} = \mathbb{E}[I_T \, e(X) \mu(d(X), X) S(X)] + \mathbb{E} \left[\frac{\partial}{\partial\epsilon} \mu_{\epsilon}(d(X), X) \bigg|_{\epsilon=0} \right],
\end{equation*}
and further write the first term on the right hand side as 
\begin{equation*}
\begin{split}
\mathbb{E}[I_T \, e(X) \mu(d(X), X) S(X)] & = \mathbb{E}[(I_T \, e(X) \mu(d(X), X) - V(d)) S(X)] \\
&= \mathbb{E} [(I_T \, e(X) \mu(d(X), X) - V(d)) S(Z)],
\end{split}
\end{equation*}
and the second term as 
\begin{equation*}
\begin{split}
& \mathbb{E} \left[\frac{\partial}{\partial\epsilon} \mu_{\epsilon}(d(X), X) \bigg|_{\epsilon=0}\right] \\
& = \mathbb{E} \left[d(X) \mathbb{E}[y(T) S(T \,|\, A, X, I_S) \,|\, A = 1, X, I_S = 1] \right. \\
& \left. \qquad + (1 - d(X)) \mathbb{E}[y(T) S(T \,|\, A, X, I_S) \,|\, A = 0, X, I_S = 1]\right] \\
& = \mathbb{E} \left[d(X) \mathbb{E}[(y(T) - \mu(1, X)) S(T \,|\, A, X, I_S) \,|\, A = 1, X, I_S = 1] \right.\\
& \left. \qquad + (1- d(X)) \mathbb{E}[(y(T) - \mu(0, X)) S(T \,|\, A, X, I_S) \,|\, A = 0, X, I_S = 1]\right] \\
& = \mathbb{E}\left[d(X) \mathbb{E}\left[\frac{I_S \, A}{\pi_S(X) \pi_A(X)}(y(T) - \mu(1, X)) S(T \,|\, A, X, I_S) \bigg| X \right] \right.\\
& \left. \qquad + (1- d(X)) \mathbb{E}\left[\frac{I_S \, (1 - A)}{\pi_S(X) (1 - \pi_A(X))} (y(T) - \mu(0, X)) S(T \,|\, A, X, I_S) \bigg| X \right] \right] \\
& = \mathbb{E}\left[\frac{I_S}{\pi_S(X)} \left( d(X) \frac{A}{\pi_A(X)}(y(T) - \mu(1, X)) \right.\right. \\
& \left.\left. \qquad + (1- d(X)) \frac{1 - A}{1 - \pi_A(X)}(y(T) - \mu(0, X)) \right) S(T \,|\, A, X, I_S) \right] \\
& = \mathbb{E} \left[\frac{I_S}{\pi_S(X)} \frac{I\{A = d(X)\}}{\pi_d(X)}(y(T) - \mu(A, X)) S(Z) \right].
\end{split}
\end{equation*}

Therefore, the efficient influence function for the full data is
\begin{equation*}
\phi_d^F = I_T \, e(X) \mu(d(X), X) + \frac{I_S}{\pi_S(X)} \frac{I\{A = d(X)\}}{\pi_d(X)}(y(T) - \mu(A, X)) - V(d).
\end{equation*}

Next, we consider the observed data $O = (X, A, U, \Delta, I_S, I_T)$ due to right censoring. According to \citet[Section 10.4]{tsiatis2006semiparametric}, the EIF based on the observed data is given by 
\begin{equation*}
\phi_d = \frac{\Delta \, \phi_d^F}{S_C(U \,|\, A, X)} + \int_0^\infty \frac{L(u, A, X)}{S_C(u \,|\, A, X)} \mathrm{d}M_C(u \,|\, A, X),
\end{equation*}
where
\begin{equation*}
\begin{split}
& L(u, A, X) = \mathbb{E}[\phi_d^F \,|\, T \geq u, A, X] \\
& = I_T \, e(X) \mu(d(X), X) + \frac{I_S}{\pi_S(X)} \frac{I\{A = d(X)\}}{\pi_d(X)}(Q(u, A, X) - \mu(A, X)) - V(d).
\end{split}
\end{equation*}

Since we have
\begin{equation}\label{eq:mt}
\begin{split}
& \int_0^\infty \frac{\mathrm{d}M_C(u \,|\, A, X)}{S_C(u \,|\, A, X)} = \int_0^\infty \frac{\mathrm{d}N_C(u)}{S_C(u \,|\, A, X)} - \int_0^U \frac{\mathrm{d}\Lambda_C(u \,|\, A, X)}{\exp\{\Lambda_C(u \,|\, A, X)\}} \\
& = 1 - \frac{\Delta}{S_C(U \,|\, A, X)},
\end{split}
\end{equation}
we conclude that 
\begin{align*}
\phi_d = & \frac{I_S}{\pi_S(X)} \frac{I\{A = d(X)\}}{\pi_d(X)} \frac{\Delta\, y(U)}{S_C(U \,|\, A,X)} - V(d) \\ 
& + \left( I_T \,e(X) - \frac{I_S}{\pi_S(X)} \frac{I\{A = d(X)\}}{\pi_d(X)}\right) \mu(d(X), X) \\
& + \frac{I_S}{\pi_S(X)} \frac{I\{A = d(X)\}}{\pi_d(X)} \int_0^\infty \frac{\mathrm{d}M_C(u\,|\,A,X)}{S_C(u\,|\,A,X)} Q(u, A, X).
\end{align*}

\section{Proof of Theorem~\ref{thm:para} and Corollary~\ref{cor:para.rmst}}

\subsection{Double robustness}
\label{pf:dr}

We start with the proof of the double robustness property. We show that EIF-based estimator is consistent when either the survival outcome model or the models for the sampling score, the propensity score and the censoring process are correctly specified. Under some regularity conditions, the nuisance estimators $\hat{\mu}(a, x)$, $\hat{Q}(u,a,x)$, $\hat{\pi}_S (x)$, $\hat{\pi}_A (x)$ and $\hat{S}_C (t \,|\, a, x)$ converge in probability to $\mu^\ast (a, x)$, $Q^\ast (u,a,x)$, $\pi_S^\ast (x)$, $\pi_A^\ast (x)$ and $S_C^\ast (t \,|\, a, x)$, respectively. It suffices to show that $\mathbb{E}[V^\ast (d)] = V(d)$, where
\begin{equation*}
\begin{split}
V^\ast (d) = & I_T \, e(X) \mu^\ast(A = d(X), X) \\
& + \frac{I_S}{\pi^\ast_S(X)} \frac{I\{A = d(X)\}}{\pi^\ast_d(X)} \left\{\frac{\Delta \, y(U)}{S_C^\ast(U \,|\, A, X)} - \mu^\ast(A, X) \right. \\
& \left. + \int_0^\infty \frac{\mathrm{d}M_C^\ast (u \,|\, A, X)}{S_C^\ast(u \,|\, A, X)} Q^\ast (u, A, X) \right\} \\
= & (I) + (II) + (III).
\end{split}
\end{equation*}

First, consider the case when the survival outcome model is correct, thus we have
\begin{equation*}
(I) = \mathbb{E}[I_T \, e(X) \mu^\ast(A = d(X), X)] = V(d)
\end{equation*}
By Equation~\ref{eq:mt}, we obtain
\begin{equation*}
\begin{split}
& (II) + (III) \\
& = \frac{I_S}{\pi^\ast_S(X)} \frac{I\{A = d(X)\}}{\pi^\ast_d(X)} \bigg\{y(T) - \mu^\ast(A, X) - \int_0^\infty \frac{\mathrm{d}M_C^\ast (u \,|\, A, X)}{S_C^\ast(u \,|\, A, X)} (y(T) - Q^\ast (u, A, X)) \bigg\}.
\end{split}
\end{equation*}

In this case, we have 
\begin{equation*}
\begin{split}
& \mathbb{E}\left[\frac{I_S}{\pi^\ast_S(X)} \frac{I\{A = d(X)\}}{\pi^\ast_d(X)} (y(T) - \mu^\ast(A, X))\right] \\
& = \mathbb{E}\left[ \mathbb{E}\left[\frac{I_S}{\pi^\ast_S(X)} \frac{I\{A = d(X)\}}{\pi^\ast_d(X)} (y(T) - \mu^\ast(A, X)) \,\bigg|\, X \right] \right] \\
& = \mathbb{E}\left[ \mathbb{E}\left[\mathbb{E}\left[\frac{I_S}{\pi^\ast_S(X)} \frac{I\{A = d(X)\}}{\pi^\ast_d(X)} (y(T) - \mu^\ast(A, X)) \,\bigg|\, A, X, I_S = 1 \right] \,\bigg|\, X \right] \right] \\
& = \mathbb{E}\left[ \mathbb{E}\left[\frac{I_S}{\pi^\ast_S(X)} \frac{I\{A = d(X)\}}{\pi^\ast_d(X)}\mathbb{E}[(y(T) - \mu^\ast(A, X)) \,|\, A, X, I_S = 1] \,\bigg|\, X \right] \right] \\
&= \mathbb{E}\left[ \mathbb{E}\left[\frac{I_S}{\pi^\ast_S(X)} \frac{I\{A = d(X)\}}{\pi^\ast_d(X)} (\mathbb{E}[y(T) \,|\, A, X, I_S = 1] - \mu^\ast(A, X)) \,\bigg|\, X \right] \right] = 0.
\end{split}
\end{equation*}

Also define $\mathrm{d}\tilde{M}_C(u \,|\, A, X) = \mathrm{d}\tilde{N}_C(u) - I\{C \geq u\} \mathrm{d}\Lambda_C(u \,|\, A, X)$ where $\tilde{N}_C(u) = I\{C \leq u\}$, so we have
\begin{equation*}
\begin{split}
& \mathbb{E}\left[\frac{I_S}{\pi^\ast_S(X)} \frac{I\{A = d(X)\}}{\pi^\ast_d(X)} \int_0^\infty \frac{\mathrm{d}M_C^\ast (u \,|\, A, X)}{S_C^\ast(u \,|\, A, X)} (y(T) - Q^\ast (u, A, X)) \right] \\
& = \mathbb{E}\left[\frac{I_S}{\pi^\ast_S(X)} \frac{I\{A = d(X)\}}{\pi^\ast_d(X)} \int_0^\infty \frac{\mathrm{d}\tilde{M}_C(u \,|\, A, X)}{S_C^\ast(u \,|\, A, X)} I\{T \geq u\} (y(T) - Q^\ast (u, A, X)) \right] \\
& = \mathbb{E}\left[\mathbb{E}\left[\frac{I_S}{\pi^\ast_S(X)} \frac{I\{A = d(X)\}}{\pi^\ast_d(X)} \int_0^\infty \frac{\mathrm{d}\tilde{M}_C^ (u \,|\, A, X)}{S_C^\ast(u \,|\, A, X)} I\{T \geq u\} (y(T) - Q^\ast (u, A, X)) \,\bigg|\, X \right]\right] \\
& = \mathbb{E}\left[\mathbb{E}\left[\mathbb{E}\left[\frac{I_S}{\pi^\ast_S(X)} \frac{I\{A = d(X)\}}{\pi^\ast_d(X)} \int_0^\infty \frac{\mathrm{d}\tilde{M}_C(u \,|\, A, X)}{S_C^\ast(u \,|\, A, X)} I\{T \geq u\} \right.\right.\right. \\
& \left.\left.\left. \qquad\qquad\qquad (y(T) - Q^\ast (u, A, X)) \,\bigg|\, A, X, C, I_S = 1 \right] \,\bigg|\, X \right]\right] \\
& = \mathbb{E}\left[\mathbb{E}\left[\frac{I_S}{\pi^\ast_S(X)} \frac{I\{A = d(X)\}}{\pi^\ast_d(X)} \int_0^\infty \frac{\mathrm{d}\tilde{M}_C(u \,|\, A, X)}{S_C^\ast(u \,|\, A, X)} \mathbb{E}\left[I\{T \geq u\} \right.\right.\right. \\
& \left.\left.\left. \qquad\qquad (y(T) - Q^\ast (u, A, X)) \,\bigg|\, A, X, C, I_S = 1 \right] \,\bigg|\, X \right]\right] \\
& = \mathbb{E}\left[\mathbb{E}\left[\frac{I_S}{\pi^\ast_S(X)} \frac{I\{A = d(X)\}}{\pi^\ast_d(X)} \int_0^\infty \frac{\mathrm{d}\tilde{M}_C(u \,|\, A, X)}{S_C^\ast(u \,|\, A, X)} \left(\mathbb{E}[I\{T \geq u\} y(T) \,|\, A, X, I_S = 1] \right.\right.\right. \\
& \left.\left.\left. \qquad\qquad - \mathbb{E}[I\{T \geq u\} \,|\, A, X, I_S = 1] Q^\ast (u, A, X) \right) \,\bigg|\, X \right]\right] = 0.
\end{split}
\end{equation*}

Next, consider the case when the models for the sampling score, the propensity score and the censoring process are correctly specified. Rearranging the terms of $V^\ast (d)$, we obtain
\begin{equation*}
\begin{split}
V^\ast (d) = & \frac{I_S}{\pi^\ast_S(X)} \frac{I\{A = d(X)\}}{\pi^\ast_d(X)} \frac{\Delta \, y(U)}{S_C^\ast(U \,|\, A, X)} \\
& + \left(I_T \, e(X) - \frac{I_S}{\pi^\ast_S(X)} \right) \mu^\ast(A = d(X), X) \\
& + \frac{I_S}{\pi^\ast_S(X)} \frac{I\{A = d(X)\}}{\pi^\ast_d(X)} \int_0^\infty \frac{\mathrm{d}M_C^\ast (u \,|\, A, X)}{S_C^\ast(u \,|\, A, X)} Q^\ast (u, A, X) \\
= & (I) + (II) + (III).
\end{split}
\end{equation*}

In this case, we have 
\begin{equation*}
(I) = \mathbb{E}\left[\frac{I_S}{\pi^\ast_S(X)} \frac{I\{A = d(X)\}}{\pi^\ast_d(X)} \frac{\Delta\, y(U)}{S_C^\ast(U \,|\, A, X)}\right] = V(d),
\end{equation*}
\begin{equation*}
\begin{split}
(II) & = \mathbb{E}\left[\left(I_T \, e(X) - \frac{I_S}{\pi^\ast_S(X)} \right) \mu^\ast(A = d(X), X) \right] \\
& = \mathbb{E}\left[\mathbb{E}\left[I_T \, e(X) - \frac{I_S}{\pi^\ast_S(X)} \bigg| X\right] \mu^\ast(A = d(X), X)\right] = 0,
\end{split}
\end{equation*}
and $(III)$ is a stochastic integral with respect to the martingale $M_C^\ast(u \,|\, A, X)$, thus equals $0$ as well, which completes the double robustness property.

\subsection{Asymptotic properties}
\label{subsec:asymp}

To establish the asymptotic results, we need some regularity conditions such that the nuisance estimators $\mu(a, x; \hat{\beta}_a,\hat{\Lambda}_{0,a})$, $Q(u,a,x; \hat{\beta}_a,\hat{\Lambda}_{0,a})$, $\pi_S (x; \hat{\lambda})$, $\pi_A (x; \hat{\theta})$ and $S_C (u \,|\, a, x; \hat{\alpha}_a,\hat{\Lambda}_{C0,a})$ converge in probability to $\mu (a, x; \beta_a^\ast,\Lambda_{0,a}^\ast)$, $Q (u,a,x; \beta_a^\ast,\Lambda_{0,a}^\ast)$, $\pi_S (x; \lambda^\ast)$, $\pi_A (x; \theta^\ast)$ and\\ $S_C (t \,|\, a, x; \alpha_a^\ast,\Lambda_{C0,a}^\ast)$, respectively.
\begin{condition}\label{cond:them1}
We assume the following conditions hold: \\
(C1) $X$ is bounded almost surely. \\
(C2) The equation $\mathbb{E}\left[\left\{A - \frac{\exp(\theta^T X)}{1 + \exp(\theta^T X)}\right\} X \right] = 0$ has a unique solution $\theta^\ast$. \\
(C3) For $a = 0, 1$, the equation
\begin{equation*}
\mathbb{E}\left[\int_0^L \left(X_i - \frac{\mathbb{E}[Y_i(u) \exp(\beta_a^T X) X]}{\mathbb{E}[Y_i(u) \exp(\beta_a^T X)]} \right) \times \mathrm{d}N_i(u)\right] = 0,
\end{equation*}
has a unique solution $\beta_a^\ast$, where $L > u$ is a pre-specified time point such that $Pr(U_i > L) > 0$. Moreover, let 
\begin{equation*}
\Lambda_{0,a}^\ast (u) = \mathbb{E}\left[\int_0^u \frac{\mathrm{d}N_i (u)}{\mathbb{E}[Y_i(u) \exp(\beta_a^{\ast T} X_i)]}\right],
\end{equation*}
and assume $\Lambda_{0,a}^\ast (L) < \infty$. \\
(C4) For $a = 0, 1$, the equation
\begin{equation*}
\mathbb{E}\left[\int_0^L \left(X_i - \frac{\mathbb{E}[Y_i(u)\exp(\alpha_a^T X) X]}{\mathbb{E}[Y_i(u) \exp(\alpha_a^T X)]}\right) \times \mathrm{d}N_i(u)\right] = 0,
\end{equation*}
has a unique solution $\alpha_a^\ast$. Moreover, let 
\begin{equation*}
\Lambda_{C0,a}^\ast (u) = \mathbb{E}\left[\int_0^u \frac{\mathrm{d}N_i (u)}{\mathbb{E}[Y_i(u) \exp(\alpha_a^{\ast T} X_i)]}\right],
\end{equation*}
and assume $\Lambda_{C0,a}^\ast (L) < \infty$. \\
(C5) The estimating equation for the sampling score model $\pi_S (X;\lambda)$ has a unique solution $\lambda^\ast$, and achieves root-$n$ rate of convergence. 
\end{condition}

Under Condition~\ref{cond:them1}, we have the following asymptotic representations:
\begin{align*}
\sqrt{n} (\hat{\theta} - \theta^\ast) &= \frac{1}{\sqrt{n}} \sum_{i=1}^n \phi_{\theta i} + o_p(1), & \sqrt{n} (\hat{\lambda} - \lambda^\ast) &= \frac{1}{\sqrt{n}} \sum_{i=1}^n \phi_{\lambda i} + o_p(1), & \\
\sqrt{n} (\hat{\beta}_a - \beta_a^\ast) &= \frac{1}{\sqrt{n}} \sum_{i=1}^n \phi_{\beta_a i} + o_p(1), & \sqrt{n} (\hat{\alpha}_a - \alpha_a^\ast) &= \frac{1}{\sqrt{n}} \sum_{i=1}^n \phi_{\alpha_a i} + o_p(1), & \text{for } a = 0, 1.
\end{align*}

We focus on the estimation of survival functions by our proposed method:
\begin{align*}
\hat{S}(t;\eta) = & \frac{1}{N}\sum_{i=1}^{N} \Bigg[I_{T,i} \, e(X_i) \hat{S}(t \,|\, A = d_\eta(X_i), X_i) \\
& \qquad\qquad + \frac{I_{S,i} I\{A_i = d_\eta(X_i)\}}{\hat{\pi}_S (X_i) \hat{\pi}_{d}(X_i)} \left\{\frac{\Delta_i \, Y_i(t)}{\hat{S}_C(t \,|\, A_i, X_i)} - \hat{S}(t \,|\, A_i, X_i) \right. \\
& \qquad\qquad \left. + \int_0^\infty \frac{\hat{S}(t \,|\, A_i, X_i) \mathrm{d}\hat{M}_C(u \,|\, A_i, X_i)}{\hat{S}(u \,|\, A_i, X_i) \hat{S}_C(u \,|\, A_i, X_i)} \right\} \Bigg],
\end{align*}
and for the ease of notation, define
\begin{equation*}
\hat{J} (t,a,x) = \frac{\Delta_i \, Y_i(t)}{\hat{S}_C(t \,|\, a, x)} - \hat{S}(t \,|\, a, x) + \int_0^\infty \frac{\hat{S}(t \,|\, a, x) \mathrm{d}\hat{M}_C(u \,|\, a, x)}{\hat{S}(u \,|\, a, x) \hat{S}_C(u \,|\, a, x)},
\end{equation*}
\begin{equation*}
J^\ast (t,a,x) = \frac{\Delta_i \, Y_i(t)}{S^\ast_C(t \,|\, a, x)} - S^\ast(t \,|\, a, x) + \int_0^\infty \frac{S^\ast(t \,|\, a, x) \mathrm{d}M_C^\ast(u \,|\, a, x)}{S^\ast(u \,|\, a, x) S_C^\ast(u \,|\, a, x)}.
\end{equation*}

Our proof has three main parts as follows.

{\bf PART 1.} By the double robustness property shown in Section~\ref{pf:dr}, we have, by the strong law of large numbers and uniform consistency, that $\hat{S}(t;\eta) = S(t;\eta) + o_p(1)$, which proves $(i)$ of Theorem~\ref{thm:para}. Moreover, define
\begin{equation*}
S_N^\ast(t;\eta) = \frac{1}{N}\sum_{i=1}^{N} \left[I_{T,i} \, e(X_i) S^\ast(t \,|\, A = d_\eta(X_i), X_i) + \frac{I_{S,i} I\{A_i = d_\eta(X_i)\}}{\pi^\ast_S (X_i) \pi^\ast_{d}(X_i)} J^\ast (t, A_i, X_i) \right],
\end{equation*}
and by applying the Taylor expansion and the counting processes result in Section~\ref{subsec:cp.cox}, we obtain
\begin{align*}
\hat{S}(t;\eta) = & S_n^\ast(t;\eta) + H_\lambda^T (\hat{\lambda} - \lambda^\ast) + H_\theta^T (\hat{\theta} - \theta^\ast) + H_{\beta_0}^T (\hat{\beta}_0 - \beta_0^\ast) + H_{\beta_1}^T (\hat{\beta}_1 - \beta_1^\ast) \\
& + H_{\alpha_0}^T (\hat{\alpha}_0 - \alpha_0^\ast) + H_{\alpha_1}^T (\hat{\alpha}_1 - \alpha_1^\ast) + R_S + o_p(N^{-1/2}),
\end{align*}
where 
\begin{equation*}
H_\lambda = \lim_{N \rightarrow \infty} \frac{1}{N} \sum_{i=1}^N \frac{\partial \hat{S}(t;\eta)}{\partial \lambda^\ast}, H_\theta = \lim_{N \rightarrow \infty} \frac{1}{N} \sum_{i=1}^N \frac{\partial \hat{S}(t;\eta)}{\partial \theta^\ast},
\end{equation*}
\begin{equation*}
\begin{split}
H_{\beta_a} = & \lim_{N \rightarrow \infty} \frac{1}{N} \sum_{i=1}^N \left\{ I_{T,i} \, e(X_i) (-1)^{a+1} G(t, a, X_i) + \frac{I_{S,i} I\{A_i = a\}}{\pi^\ast_S(X_i) \pi^\ast_d(X_i)} \left(\int_0^\infty\frac{G(t,a,X_i)\mathrm{d}M^\ast_C(u \,|\, a, X_i)}{S^\ast(u \,|\, a, X_i)S_C^\ast(u \,|\, a, X_i)} \right.\right. \\ 
& \qquad\qquad\qquad \left.\left. - G(t,a,X_i) - \int_0^\infty\frac{G(u,a,X_i)S^\ast(t \,|\, a, X_i)\mathrm{d}M^\ast_C(u \,|\, a, X_i)}{S^{\ast 2}(u \,|\, a, X_i) S^\ast_C(u \,|\, a, X_i)} \right)\right\},
\end{split}
\end{equation*}
\begin{equation*}
\begin{split}
H_{\alpha_a} = & \lim_{N \rightarrow \infty} \frac{1}{N} \sum_{i=1}^N \frac{I_{S,i} I\{A_i=a\}}{\pi^\ast_S(X_i) \pi^\ast_d(X_i)} \left\{ \frac{-\Delta_i Y_i(t)}{S^\ast_C(t \,|\, a, X_i)} G_C(t, a, X_i) \right. \\ 
& \qquad\qquad\quad \left. - \int_0^\infty\frac{G_C(u,a,X_i) S^\ast(t \,|\, a, X_i) \mathrm{d}M^\ast_C(u \,|\, a, X_i)}{S_C^{\ast 2}(u \,|\, a, X_i) S^\ast(u \,|\, a, X_i)} + \tilde{G}_C(t,a,X_i)\right\},
\end{split}
\end{equation*}
\begin{equation*}
\begin{split}
R_S & = \frac{1}{N} \sum_{i=1}^N \sum_{a=0,1} \bigg\{ I_{T,i} \, e(X_i) (-1)^{a+1} H(t,a,X_i) \\
& \qquad + \frac{I_{S,i} I\{A_i = a\}}{\pi^\ast_S(X_i) \pi^\ast_d(X_i)} \bigg(\int_0^\infty\frac{H(t,a,X_i)\mathrm{d}M^\ast_C(u \,|\, a, X_i)}{S^\ast_C(u \,|\, a, X_i)S^\ast(u \,|\, a, X_i)} - H(t,a,X_i) \\
& \qquad - \int_0^\infty\frac{H(u,a,X_i)S^\ast(t \,|\, a, X_i)\mathrm{d}M^\ast_C(u \,|\, a, X_i)}{S^\ast_C(u\,|\, a, X_i)S^{\ast 2}(u\,|\, a, X_i)} - \frac{\Delta_i Y_i(t)}{S^\ast_C(t|a,X_i)} H_C(t,a,X_i) \\
& \qquad - \int_0^\infty\frac{H_C(u,a,X_i)S^\ast(t \,|\, a, X_i)\mathrm{d}M^\ast_C(u \,|\, a, X_i)}{S_C^{\ast 2}(u \,|\, a, X_i)S^\ast(u \,|\, a, X_i)} - \tilde{H}_C(t, a, X_i) \bigg) \bigg\} \\
& = \frac{1}{N} \sum_{i=1}^N \phi_{Rs,i},
\end{split} 
\end{equation*}
with
\begin{equation*}
G(t, a, x) = -S^\ast(t \,|\, a, x)\Lambda_{0,a}^\ast(t) x^T + S^\ast(t \,|\, a, x) \exp(\beta_a^{\ast T} x) \left\{\int_0^t e_a(\beta_a^\ast,u)\mathrm{d}\Lambda_{0,a}^\ast(u) \right\}^T,
\end{equation*}
\begin{equation*}
H(t, a, x) = -S^\ast(t \,|\, a, x) \exp(\beta_a^{\ast T} x) \int_0^t\frac{n_a^{-1}\sum_{i=1}^n I\{A_i=a\} \mathrm{d}M_{a,i}(u)}{U_a^{(0)}(\beta_a^\ast,u)},
\end{equation*}
\begin{equation*}
G_C(t, a, x) = -S^\ast(t \,|\, a, x) \Lambda^\ast_{0,a}(t) x^T + S^\ast(t \,|\, a, x) \exp(\beta_a^{\ast T} x) \left\{\int_0^t e_a(\beta_a^\ast,u)\mathrm{d}\Lambda_{0,a}^\ast(u) \right\}^T,
\end{equation*}
\begin{equation*}
H_C(t, a, x) = -S^\ast(t \,|\, a, x) \exp(\beta_a^{\ast T} x) \int_0^t\frac{n_a^{-1}\sum_{i=1}^n I\{A_i=a\} \mathrm{d}M_{a,i}(u)}{U_a^{(0)}(\beta_a^\ast,u)},
\end{equation*}
\begin{equation*}
\tilde{G}_C(t, a, x) = \int_0^{U_i}\frac{S^\ast(t \,|\, a, x)\mathrm{d}\Lambda_{C}^\ast(u \,|\, a, x)}{S^\ast_C(u \,|\, a, x) S^\ast(u \,|\, a, x)} x^T + \left\{\int_0^t \frac{S^\ast(t \,|\, a, x) e_a(\beta_a^\ast,u)\mathrm{d}\Lambda_{0,a}^\ast(u)}{S^\ast_C(u \,|\, a, x) S^\ast(u \,|\, a, x)} \right\}^T,
\end{equation*}
\begin{equation*}
\tilde{H}_C(t, a, x) = \int_0^t \frac{S^\ast(t \,|\, a, x) n_a^{-1}\sum_{i=1}^n I\{A_i=a\} \mathrm{d}M_{a,i}(u)}{S^\ast_C(u \,|\, a, x) S^\ast(u \,|\, a, x) U_a^{(0)}(\beta_a^\ast,u)}.
\end{equation*}
Thus, we have
\begin{equation}\label{eq:asym.dist}
\sqrt{N} \left\{\hat{S}(t;\eta) - S(t;\eta)\right\} = \frac{1}{\sqrt{N}} \sum_{i=1}^N (\xi_{1,i}(t;\eta) + \xi_{2,i}(t;\eta)) + o_p(1),
\end{equation}
where
\begin{equation*}
\xi_{1,i}(t;\eta) = S_n^\ast(t;\eta) - S(t;\eta),
\end{equation*}
\begin{equation*}
\xi_{2,i}(t;\eta) = H_\lambda^T \phi_{\lambda^\ast, i} + H_\theta^T \phi_{\theta^\ast,i} + \sum_{a=0,1} H_{\beta_a}^T \phi_{\beta_0^\ast,i} + \sum_{a=0,1} H_{\alpha_a}^T \phi_{\alpha_a^\ast,i} + H_{\alpha_1}^T + \phi_{Rs,i},
\end{equation*}
and $\xi_{1,i}(t;\eta), \xi_{2,i}(t;\eta)$ are independent mean-zero processes. Therefore, we obtain that $\sqrt{N} \left\{\hat{S}(t;\eta) - S(t;\eta)\right\}$ converges weakly to a mean-zero Gaussian process, which proves $(ii)$ of Theorem~\ref{thm:para}.

{\bf PART 2.} We show that $N^{1/3} \|\hat{\eta} - \eta^\ast\|_2 = O_p(1)$. Recall that
\begin{equation*}
\hat{\eta} = \arg\max_\eta \hat{S}(t;\eta) \text{ and } \eta^\ast = \arg\max_\eta S(t;\eta).
\end{equation*}

By Assumption~\ref{asmp:regu} $(i)$, $S(t;\eta)$ is twice continuously differentiable at a neighborhood of $\eta^\ast$; in Step 1, we show that $\hat{S}(t;\eta) = S(t;\eta) + o_p(1), \forall \eta$; since $\hat{\eta}$ maximizes $\hat{S}(t;\eta)$, we have that $\hat{S}(t;\hat{\eta}) \geq \sup_{\eta} \hat{S}(t;\eta)$, thus by the Argmax theorem, we have $\hat{\eta} \overset{p}{\to} \eta^\ast$ as $N \to \infty$.

In order to establish the $N^{-1/3}$ rate of convergence of $\hat{\eta}$, we apply Theorem 14.4 (Rate of convergence) of \citet{kosorok2008introduction}, and need to find the suitable rate that satisfies three conditions below.

{\bf Condition 1} For every $\eta$ in a neighborhood of $\eta^\ast$ such that $\|\eta - \eta^\ast\|_2 < \delta$, by Assumption~\ref{asmp:regu} $(i)$, we apply the second-order Taylor expansion,
\begin{align*}
S(t;\eta) - S(t;\eta^\ast) & = S'(\eta^\ast) \|\eta - \eta^\ast\|_2 + \frac{1}{2} S''(\eta^\ast) \|\eta - \eta^\ast\|_2^2 + o(\|\eta - \eta^\ast\|_2^2) \\
& = \frac{1}{2} S''(\eta^\ast)\|\eta - \eta^\ast\|_2^2 + o(\|\eta - \eta^\ast\|_2^2),
\end{align*}
and as $S''(\eta^\ast) < 0$, there exists $c_0 = -\frac{1}{2} S''(\eta^\ast) > 0$ such that $S(t;\eta) - S(t;\eta^\ast) \leq -c_0 \|\eta - \eta^\ast\|_2^2$. 

{\bf Condition 2} For all $N$ large enough and sufficiently small $\delta$, we consider the centered process $\hat{S} - S$, and have that 
\begin{align*}
& \mathbb{E} \left[ \sqrt{N} \sup_{\|\eta - \eta^\ast\|_2 < \delta} \left| \hat{S}(t;\eta) - S(t;\eta) - \left\{\hat{S}(t;\eta^\ast) - S(t;\eta^\ast)\right\} \right| \right] \\
& = \mathbb{E} \Bigg[ \sqrt{N} \sup_{\|\eta - \eta^\ast\|_2 < \delta} \left| \hat{S}(t;\eta) - S_n^\ast(t;\eta) + S_n^\ast(t;\eta) - S(t;\eta) \right. \\
& \left. \qquad\quad - \left\{ \hat{S}(t;\eta^\ast) - S_n^\ast(t;\eta^\ast) + S_n^\ast(t;\eta^\ast) - S(t;\eta^\ast)\right\} \right| \Bigg] \\
& \leq \mathbb{E} \left[ \sqrt{N} \sup_{\|\eta - \eta^\ast\|_2 < \delta} \left| \hat{S}(t;\eta) - S_n^\ast(t;\eta) - \left\{ \hat{S}(t;\eta^\ast) - S_n^\ast(t;\eta^\ast) \right\} \right| \right] \tag{$I$} \\
& \quad + \mathbb{E} \left[ \sqrt{N} \sup_{\|\eta - \eta^\ast\|_2 < \delta} \left| S_n^\ast(t;\eta) - S(t;\eta) - \left\{ S_n^\ast(t;\eta^\ast) - S(t;\eta^\ast)\right\} \right| \right], \tag{$II$}
\end{align*}
and we bound $(I)$ and $(II)$ respectively as follows.

{\bf Condition 2.1} To bound $(II)$, we need the useful facts that
\begin{equation*}
I\{A = d_\eta(X)\} - I\{A = d_{\eta\ast}(X)\} = (2A - 1)(d_\eta(X) - d_{\eta\ast}(X)),
\end{equation*}
\begin{equation*}
S^\ast(t \,|\, d_\eta(X_i), X_i) - S^\ast(t \,|\, d_{\eta^\ast}(X_i), X_i) = (S^\ast(t \,|\, 1, X_i) - S^\ast(t \,|\, 0, X_i))(d_\eta(X_i) - d_{\eta\ast}(X_i)),
\end{equation*}
and obtain 
\begin{align*}
& S_n^\ast (t;\eta) - S_n^\ast (t;\eta^\ast) = \frac{1}{N} \sum_{i=1}^{N} (d_\eta(X_i) - d_{\eta\ast}(X_i)) \\
& \quad \times \left\{I_{T,i} \, e(X_i) (S^\ast(t \,|\, 1, X_i) - S^\ast(t \,|\, 0, X_i)) + \frac{(2A_i - 1) I_{S,i}}{\pi_S^\ast(X_i) \pi_d^\ast(X_i)} J^\ast(t, A_i, X_i) \right\}.
\end{align*}

Define a class of functions
\begin{align*}
\mathcal{F}_\eta^1 = & \bigg\{ (d_\eta(x) - d_{\eta\ast}(x)) \bigg(I_T \, e(x) (S^\ast(t \,|\, 1, x) - S^\ast(t \,|\, 0, x)) + \frac{(2a - 1) I_S}{\pi_a^\ast(x) \pi_S^\ast (x)} J^\ast(t,a,x) \bigg): \\
& \qquad \|\eta - \eta^\ast\|_2 < \delta \bigg\},
\end{align*}
and let $M_1 = \sup \left| I_T \, e(x) (S^\ast(t \,|\, 1, x) - S^\ast(t \,|\, 0, x)) + \frac{(2a - 1) I_S}{\pi_a^\ast(x) \pi_S^\ast (x)} J^\ast(t,a,x)\right|$. By Assumption~\ref{asmp:cnpc}, \ref{asmp:psi} and Condition~\ref{cond:them1}, we have that $M_1 < \infty$. 

When $\|\eta -\eta^\ast\|_2 < \delta$, by Condition~\ref{cond:them1} \textit{(C1)}, there exists a constant $0 < k_0 < \infty$ such that $|(1,x^T) (\eta - \eta^\ast)| < k_0 \delta$; furthermore, we show that $|d_\eta(x) - d_{\eta\ast}(x)| = |I\{(1,x^T) \eta > 0\} - I\{(1,x^T) \eta^\ast > 0\}| \leq I\{-k_0\delta\leq (1,x^T)\eta^\ast\leq k_0\delta\}$, by considering the three cases:
\begin{itemize}
\item when $-k_0\delta \leq (1,x^T)\eta^\ast \leq k_0\delta$, we have $|d_\eta(x) - d_{\eta\ast}(x)| \leq 1 = I\{-k_0\delta\leq (1,x^T)\eta^\ast\leq k_0\delta\}$;
\item when $(1,x^T)\eta^\ast > k_0 \delta > 0$, we have $(1,x^T)\eta = (1,x^T)(\eta - \eta^\ast) + (1,x^T)\eta^\ast > 0$, so $|d_\eta(x) - d_{\eta\ast}(x)| = 0 = I\{-k_0\delta\leq (1,x^T)\eta^\ast\leq k_0\delta\}$; 
\item when $(1,x^T)\eta^\ast < -k_0 \delta < 0$, we have $(1,x^T)\eta = (1,x^T)(\eta - \eta^\ast) + (1,x^T)\eta^\ast < 0$, so $|d_\eta(x) - d_{\eta\ast}(x)| = 0 = I\{-k_0\delta\leq (1,x^T)\eta^\ast\leq k_0\delta\}$.
\end{itemize}

Thus we can define the envelope of $\mathcal{F}_\eta^1$ as $F_1 = M_1 I\{-k_0 \delta \leq (1, x^T) \eta^\ast \leq k_0 \delta\}$. By Assumption~\ref{asmp:regu} \textit{(ii)}, there exists a constant $0 < k_1 < \infty$ such that
\begin{equation*}
\|F_1\|_{P,2} \leq M_1 \sqrt{Pr(-k_0 \delta \leq (1, x^T) \eta^\ast \leq k_0 \delta)} \leq M_1 \sqrt{2k_0 k_1} \delta^{1/2}.
\end{equation*}

By Lemma 9.6 and Lemma 9.9 of \cite{kosorok2008introduction}, we have that $\mathcal{F}_\eta^1$, a class of indicator functions, is a Vapnik-Cervonenkis (VC) class with bounded bracketing entropy $J_{[]}^\ast(1,\mathcal{F}_\eta^1) < \infty$.

Since we have the fact that
\begin{align*}
\mathbb{G}_N \mathcal{F}_\eta^1 & = N^{-1/2} \sum_{i=1}^N \left\{\mathcal{F}_\eta^1 - \mathbb{E}[\mathcal{F}_\eta^1] \right\} \\
& = \sqrt{N} \left(S_n^\ast(t;\eta) - S_n^\ast(t;\eta^\ast) - \left\{S(t;\eta) - S(t;\eta^\ast)\right\}\right),
\end{align*}
By Theorem 11.2 of \cite{kosorok2008introduction}, we obtain that there exists a constant $0 < c_1 < \infty$,
\begin{equation*}
(II) = \mathbb{E}\left[\sup_{\|\eta - \eta^\ast\|_2 < \delta} |\mathbb{G}_N \mathcal{F}_\eta^1|\right] \leq c_1 J_{[]}^\ast(1,\mathcal{F}_\eta^1) \|F_1\|_{P,2} \leq c_1 J_{[]}^\ast(1,\mathcal{F}_\eta^1)M_1\sqrt{2k_0 k_1} \delta^{1/2} = \tilde{c}_1 \delta^{1/2},
\end{equation*}
so we conclude that $(II) \leq \tilde{c}_1 \delta^{1/2}$ where $\tilde{c}_1 > 0$ is a finite constant.

{\bf Condition 2.2} To bound $(I)$, first we have 
\begin{align*}
& \hat{S}(t;\eta) - S_n^\ast(t;\eta) - \{\hat{S}(t;\eta^\ast) - S_n^\ast(t;\eta^\ast)\} = \hat{S}(t;\eta) - \hat{S}(t;\eta^\ast) - \{S_n^\ast(t;\eta) - S_n^\ast(t;\eta^\ast)\} \\
& = \frac{1}{N} \sum_{i=1}^{N} (d_\eta(X_i) - d_{\eta\ast}(X_i)) \left[I_{T,i} \, e(X_i) \{\hat{S}(t|1,X_i) - \hat{S}(t|0,X_i) - (S^\ast(t|1,X_i) - S^\ast(t|0,X_i)) \} \right. \\
& \left. \quad + \frac{(2A_i - 1) I_{S,i}}{\hat{\pi}_{A_i}(X_i) \hat{\pi}_S(X_i)} \hat{J}(t,A_i,X_i) - \frac{(2A_i - 1) I_{S,i}}{\pi_{A_i}^\ast(X_i) \pi^\ast_S(X_i)} J^\ast(t,A_i,X_i) \right],
\end{align*}
and then apply the Taylor expansion and counting processes result in Section~\ref{subsec:cp.cox},
\begin{equation}\label{eq:cond2.2}
\begin{split}
& \hat{S}(t;\eta) - S_n^\ast(t;\eta) - \{\hat{S}(t;\eta^\ast) - S_n^\ast(t;\eta^\ast)\} \\
& = \frac{1}{N} \sum_{i=1}^{N} (d_\eta(X_i) - d_{\eta\ast}(X_i)) \times \left\{ D_\lambda (\hat{\lambda} - \lambda^\ast) + D_\theta (\hat{\theta} - \theta^\ast) + D_{\beta_0} (\hat{\beta}_0 - \beta_0^\ast) \right. \\
& \left. \quad + D_{\beta_1} (\hat{\beta}_1 - \beta_1^\ast) + D_{\alpha_0} (\hat{\alpha}_0 - \alpha_0^\ast) + D_{\alpha_1} (\hat{\alpha}_1 - \alpha_1^\ast) + R_{S,i} \right\} + o_p(N^{-1/2}),
\end{split}
\end{equation}
where
\begin{align*}
D_\lambda = -\frac{(2A_i - 1)I_{S,i}}{\pi^\ast_{A_i}(X_i) \pi^{\ast 2}_S(X_i)} J^\ast(t,A_i,X_i) \left(\frac{\partial \pi_S^\ast(X_i)}{\partial\lambda}\right)^T,
\end{align*}
\begin{align*}
D_\theta = -\frac{I_{S,i}}{\pi_{A_i}^{\ast 2}(X_i) \pi^\ast_S(X_i)} J^\ast(t,A_i,X_i) \left(\frac{\partial \pi_A^\ast(X_i)}{\partial\theta}\right)^T,
\end{align*}
\begin{align*}
D_{\beta_a} = & I_{T,i} \, e(X_i) (-1)^{a+1} G(t, a, X_i) + \frac{(2A_i - 1)I\{A_i = a\} I_{S,i}}{\pi_{A_i}^\ast(X_i) \pi_S^\ast(X_i)} \left(\int_0^\infty\frac{G(t,a,X_i)\mathrm{d}M_C^\ast(u \,|\, a, X_i)}{S_C^\ast(u \,|\, a, X_i)S^\ast(u \,|\, a, X_i)} \right. \\ 
& \left. - G(t,a,X_i) - \int_0^\infty \frac{G(u,a,X_i) S^\ast(t \,|\, a, X_i) \mathrm{d}M_C^\ast(u \,|\, a, X_i)}{S_C^\ast(u \,|\, a, X_i) S^{\ast 2}(u \,|\, a, X_i)} \right),
\end{align*}
\begin{align*}
D_{\alpha_a} = & \frac{(2A_i - 1) I\{A_i = a\} I_{S,i}}{\pi^\ast_{A_i}(X_i) \pi_S^\ast(X_i)} \left\{-\frac{\Delta_i \, Y_i(t)}{S^\ast_C(t \,|\, a, X_i)} G_C(t,a,X_i) \right. \\ 
& \left. \quad - \int_0^\infty \frac{G_C(u,a,X_i) S^\ast(t \,|\, a, X_i) \mathrm{d}M_C^\ast(u \,|\, a, X_i)}{S_C^{\ast 2}(u \,|\, a, X_i) S^\ast(u \,|\, a, X_i)} + \tilde{G}_C(t,a,X_i) \right\},
\end{align*}
\begin{align*}
R_{S,i} = & \sum_{a=0,1} \bigg[ I_{T,i} \, e(X_i) (-1)^{a+1} H(t,a,X_i) + \frac{(2A_i - 1) I\{A_i = a\} I_{S,i}}{\pi_{A_i}^\ast(X_i) \pi_S^\ast(X_i)} \bigg(\int_0^\infty\frac{H(t,a,X_i)\mathrm{d}M_C^\ast(u \,|\, a, X_i)}{S_C^\ast(u \,|\, a, X_i)S^\ast(u \,|\, a, X_i)} \\
& \quad - H(t,a,X_i) - \int_0^\infty\frac{H(u,a,X_i)S^\ast(t \,|\, a, X_i)\mathrm{d}M_C^\ast(u \,|\, a, X_i)}{S_C^\ast(u \,|\, a, X_i)S^{\ast 2}(u \,|\, a, X_i)} \\
& \quad - \frac{\Delta_i \, Y_i(t)}{S^\ast_C(t \,|\, a, X_i)} H_C(t,a,X_i) - \int_0^\infty\frac{H_C(u,a,X_i)S^\ast(t \,|\, a, X_i)\mathrm{d}M_C^\ast(u \,|\, a, X_i)}{S_C^{\ast 2}(u \,|\, a, X_i)S^\ast(u \,|\, a, X_i)} - \tilde{H}_C(t, a, X_i) \bigg) \bigg].
\end{align*}

Similarly, we define the following classes of functions:
\begin{align*}
\mathcal{F}_\eta^2 = \left\{ (d_\eta(x) - d_{\eta^\ast}(x)) \frac{(2a-1) I_{S,i}}{\pi_{a}^\ast(x) \pi_S^{\ast 2}(x)} J^\ast(t,a,x) \left(\frac{\partial \pi_S^\ast(x)}{\partial\lambda}\right)^T : \|\eta - \eta^\ast\|_2 < \delta \right\},
\end{align*}
\begin{align*}
\mathcal{F}_\eta^3 = \left\{ (d_\eta(x) - d_{\eta^\ast}(x)) \frac{-I_{S,i}}{\pi_{a}^{\ast 2}(x) \pi_S^\ast(x)} J^\ast(t,a,x) \left(\frac{\partial \pi_A^\ast(x)}{\partial\theta}\right)^T : \|\eta - \eta^\ast\|_2 < \delta \right\},
\end{align*}
\begin{align*}
\mathcal{F}_\eta^4 = \Bigg\{&(d_\eta(x) - d_{\eta^\ast}(x)) \Bigg[ I_{T} \, e(x) (-1)^{a+1} G(t, a, x) + \frac{(2a - 1) I_{S}}{\pi_{a}^\ast(x) \pi_S^\ast(x)} \\
& \quad \times \left(\int_0^\infty\frac{G(t,a,x)\mathrm{d}M_C^\ast(u \,|\, a, x)}{S_C^\ast(u \,|\, a, x)S^\ast(u \,|\, a, x)} - G(t,a,x) \right. \\
& \left. \quad - \int_0^\infty \frac{G(u,a,x) S^\ast(t \,|\, a, x) \mathrm{d}M_C^\ast(u \,|\, a, x)}{S_C^\ast(u \,|\, a, x) S^{\ast 2}(u \,|\, a, x)} \right) \Bigg] : \|\eta - \eta^\ast\|_2 < \delta \Bigg\},
\end{align*}
\begin{align*}
\mathcal{F}_\eta^5 = \Bigg\{&(d_\eta(x) - d_{\eta^\ast}(x)) \Bigg[ I_{T} \, e(x) (-1)^{a+1} G(t, a, x) + \frac{(2a - 1) I_{S}}{\pi_{a}^\ast(x) \pi_S^\ast(x)} \\
& \quad \times \left(\int_0^\infty\frac{G(t,a,x)\mathrm{d}M_C^\ast(u \,|\, a, x)}{S_C^\ast(u \,|\, a, x)S^\ast(u \,|\, a, x)} - G(t,a,x) \right. \\
& \left. \quad - \int_0^\infty \frac{G(u,a,x) S^\ast(t \,|\, a, x) \mathrm{d}M_C^\ast(u \,|\, a, x)}{S_C^\ast(u \,|\, a, x) S^{\ast 2}(u \,|\, a, x)} \right) \Bigg] : \|\eta - \eta^\ast\|_2 < \delta \Bigg\},
\end{align*}
\begin{align*}
\mathcal{F}_\eta^6 = \Bigg\{ & (d_\eta(x) - d_{\eta^\ast}(x)) \Bigg[ \frac{(2a - 1) I_{S}}{\pi_{a}^\ast(x) \pi_S^\ast(x)} \left\{-\frac{\Delta \, Y(t)}{S^\ast_C(t \,|\, a, x)} G_C(t,a,x) \right. \\ 
& \left. \quad - \int_0^\infty \frac{G_C(u,a,x) S^\ast(t \,|\, a, x) \mathrm{d}M_C^\ast(u \,|\, a, x)}{S_C^{\ast 2}(u \,|\, a, x) S^\ast(u \,|\, a, x)} + \tilde{G}_C(t,a,x) \right\} \Bigg] : \|\eta - \eta^\ast\|_2 < \delta \Bigg\},
\end{align*}
\begin{align*}
\mathcal{F}_\eta^7 = \Bigg\{ & (d_\eta(x) - d_{\eta^\ast}(x)) \Bigg[ \frac{(2a - 1) I_{S}}{\pi_{a}^\ast(x) \pi_S^\ast(x)} \left\{-\frac{\Delta \, Y(t)}{S^\ast_C(t \,|\, a, x)} G_C(t,a,x) \right. \\ 
& \left. \quad - \int_0^\infty \frac{G_C(u,a,x) S^\ast(t \,|\, a, x) \mathrm{d}M_C^\ast(u \,|\, a, x)}{S_C^{\ast 2}(u \,|\, a, x) S^\ast(u \,|\, a, x)} + \tilde{G}_C(t,a,x) \right\} \Bigg] : \|\eta - \eta^\ast\|_2 < \delta \Bigg\},
\end{align*}
\begin{align*}
\mathcal{F}_\eta^8 = \Bigg\{& (d_\eta(x) - d_{\eta^\ast}(x)) \Bigg[ \sum_{a=0,1} \bigg[ I_{T} \, e(x)^{a+1} H(t,a,x) + \frac{(2a - 1) I_{S}}{\pi_{a}^\ast(x) \pi_S^\ast(x)} \\
& \quad \times \bigg(\int_0^\infty\frac{H(t,a,x)\mathrm{d}M_C^\ast(u \,|\, a, x)}{S_C^\ast(u \,|\, a, x)S^\ast(u \,|\, a, x)} - H(t,a,x) \\
& \quad - \int_0^\infty\frac{H(u,a,x)S^\ast(t \,|\, a, x)\mathrm{d}M_C^\ast(u \,|\, a, x)}{S_C^\ast(u \,|\, a, x)S^{\ast 2}(u \,|\, a, x)} - \frac{\Delta \, Y(t)}{S^\ast_C(t \,|\, a, x)} H_C(t,a,x) \\
& \quad - \int_0^\infty\frac{H_C(u,a,x)S^\ast(t \,|\, a, x)\mathrm{d}M_C^\ast(u \,|\, a, x)}{S_C^{\ast 2}(u \,|\, a, x)S^\ast(u \,|\, a, x)} - \tilde{H}_C(t, a, x) \bigg) \bigg] \Bigg] : \|\eta - \eta^\ast\|_2 < \delta \Bigg\}.
\end{align*}

Let
\begin{equation*}
M_2 = \sup \left|\frac{(2a-1)}{\pi_{a}^\ast(x)} J^\ast(t,a,x) \left(\frac{\partial \pi_S^\ast(x)}{\partial\lambda}\right)^T\right|,
\end{equation*}
where $M_2 \in \mathbb{R}^+$ and the supremum is taken over all the coordinates; and $M_3, \ldots, M_8$ are defined accordingly for $\mathcal{F}_\eta^3, \ldots, \mathcal{F}_\eta^8$. By Assumption~\ref{asmp:cnpc}, \ref{asmp:psi} and Condition~\ref{cond:them1}, we have that $M_2, \ldots, M_8 < \infty$.

Using the same technique as in {\bf Condition 2.1}, we define the envelop of $\mathcal{F}_\eta^j$ as $F_j = M_j I\{-k_0 \delta \leq (1, x^T) \eta^\ast \leq k_0 \delta\}$ for $j = 2, \ldots, 8$, and obtain that
\begin{equation*}
\|F_j\|_{P,2} \leq \tilde{M}_j \delta^{1/2} < \infty, \quad j = 2, \ldots, 8,
\end{equation*}
where $\tilde{M}_2, \ldots, \tilde{M}_8$ are some finite constants, and that $\mathcal{F}_\eta^j$ is a VC class with bounded bracketing entropy $J_{[]}^\ast(1,\mathcal{F}_\eta^j) < \infty$, for $j = 2, \ldots, 8$. By Theorem 11.2 of \cite{kosorok2008introduction}, we obtain
\begin{equation*}
\mathbb{E}\left[\sup_{\|\eta - \eta^\ast\|_2<\delta} \left|\mathbb{G}_N \mathcal{F}_\eta^j\right| \right] \leq c_j J_{[]}^\ast(1,\mathcal{F}_\eta^j) \|F_j\|_{P,2}, \quad j = 2, \ldots, 8,
\end{equation*}
where $c_2, \ldots, c_8$ are some finite constants. That is, we have 
\begin{equation*}
\mathbb{E}\left[\sup_{\|\eta - \eta^\ast\|_2<\delta} \left| \mathbb{G}_N \mathcal{F}_\eta^8 \right| \right] \leq \tilde{c}_8 \delta^{1/2},
\end{equation*}
and furthermore by Theorem 2.14.5 of \cite{van1996weak}, we obtain
\begin{align*}
\left\{\mathbb{E} \left[\sup_{\|\eta - \eta^\ast\|_2<\delta} \|\mathbb{G}_n \mathcal{F}_\eta^j\|_2^2\right]\right\}^{1/2}
& \leq l_j \left\{\mathbb{E}\left[\sup_{\|\eta - \eta^\ast\|_2<\delta} |\mathbb{G}_n \mathcal{F}_\eta^j|\right] + \|F_j\|_{P,2}\right\} && \\
& \leq l_j \{c_j J_{[]}^\ast(1,\mathcal{F}_\eta^j) + 1\} \|F_j\|_{P,2} && \\
& \leq \tilde{c}_j \delta^{1/2}, && j = 2, \ldots, 7,
\end{align*}
where $l_2, \ldots, l_7$ and $\tilde{c}_2, \ldots, \tilde{c}_7$ are some finite constants.

By Equation~\eqref{eq:cond2.2}, we have that
\begin{align*}
(I) & = \mathbb{E}\left[N^{1/2} \sup_{\|\eta - \eta^\ast\|_2 < \delta} \left|\hat{S}(t;\eta) - S_N^\ast(t;\eta) - \{\hat{S}(t;\eta^\ast) - S_N^\ast(t;\eta^\ast)\}\right| \right] \\
& \leq \mathbb{E} \Bigg[\sup_{\|\eta - \eta^\ast\|_2 < \delta} \bigg\{ |\mathbb{G}_n \mathcal{F}_\eta^2(\hat{\lambda} - \lambda^\ast)| + |\mathbb{G}_n \mathcal{F}_\eta^3(\hat{\theta} - \theta^\ast)| + |\mathbb{G}_n \mathcal{F}_\eta^4(\hat{\beta}_0 - \beta_0^\ast)| + |\mathbb{G}_n \mathcal{F}_\eta^5(\hat{\beta}_1 - \beta_1^\ast)| \\
& \qquad + |\mathbb{G}_n \mathcal{F}_\eta^6(\hat{\alpha}_0 - \alpha_0^\ast)| + |\mathbb{G}_n \mathcal{F}_\eta^7(\hat{\alpha}_1 - \alpha_1^\ast)| + |\mathbb{G}_n \mathcal{F}_\eta^8| \bigg\} + o_p(1) \Bigg] \\
& \leq N^{-1/2} \,\Bigg\{ \mathbb{E} \left[\sup_{\|\eta - \eta^\ast\|_2 < \delta} |\mathbb{G}_n \mathcal{F}_\eta^2 \cdot N^{1/2} (\hat{\lambda} - \lambda^\ast)|\right] + \mathbb{E} \left[\sup_{\|\eta - \eta^\ast\|_2 < \delta}|\mathbb{G}_n \mathcal{F}_\eta^3 \cdot N^{1/2}(\hat{\theta} - \theta^\ast)|\right] \\
& \qquad + \mathbb{E} \left[\sup_{\|\eta - \eta^\ast\|_2 < \delta}|\mathbb{G}_n \mathcal{F}_\eta^4 \cdot N^{1/2} (\hat{\beta}_0 - \beta_0^\ast)|\right] + \mathbb{E} \left[\sup_{\|\eta - \eta^\ast\|_2 < \delta}|\mathbb{G}_n \mathcal{F}_\eta^5 \cdot N^{1/2}(\hat{\beta}_1 - \beta_1^\ast)|\right] \\
& \qquad + \mathbb{E} \left[\sup_{\|\eta - \eta^\ast\|_2 < \delta}|\mathbb{G}_n \mathcal{F}_\eta^6 \cdot N^{1/2}(\hat{\alpha}_0 - \alpha_0^\ast)|\right] + \mathbb{E} \left[\sup_{\|\eta - \eta^\ast\|_2 < \delta}|\mathbb{G}_n \mathcal{F}_\eta^7 \cdot N^{1/2}(\hat{\alpha}_1 - \alpha_1^\ast)|\right]  \Bigg\} \\
& \quad + \mathbb{E}\left[\sup_{\|\eta - \eta^\ast\|_2<\delta} \left| \mathbb{G}_N \mathcal{F}_\eta^8 \right| \right] + o_p(1),
\end{align*}
and then by the Cauchy-Schwarz inequality, we obtain
\begin{align*}
(I) \leq & \, N^{-1/2} \left\{\mathbb{E}[N \|\hat{\lambda} - \lambda^\ast\|_2^2]\right\}^{1/2} \left\{\mathbb{E} \left[\sup_{\|\eta - \eta^\ast\|_2 < \delta} \|\mathbb{G}_N \mathcal{F}_\eta^2\|_2^2\right]\right\}^{1/2} \\
& + N^{-1/2} \left\{\mathbb{E}[N \|\hat{\theta} - \theta^\ast\|_2^2]\right\}^{1/2} \left\{\mathbb{E} \left[\sup_{\|\eta - \eta^\ast\|_2<\delta} \|\mathbb{G}_N \mathcal{F}_\eta^3\|_2^2\right]\right\}^{1/2} \\
& + N^{-1/2} \left\{\mathbb{E}[N \|\hat{\beta}_0 - \beta_0^\ast\|_2^2]\right\}^{1/2} \left\{\mathbb{E} \left[\sup_{\|\eta - \eta^\ast\|_2<\delta} \|\mathbb{G}_N \mathcal{F}_\eta^4\|_2^2\right]\right\}^{1/2} \\
& + N^{-1/2} \left\{\mathbb{E}[N \|\hat{\beta}_1 - \beta_1^\ast\|_2^2]\right\}^{1/2} \left\{\mathbb{E} \left[\sup_{\|\eta - \eta^\ast\|_2<\delta} \|\mathbb{G}_N \mathcal{F}_\eta^5\|_2^2\right]\right\}^{1/2} \\
& + N^{-1/2} \left\{\mathbb{E}[N \|\hat{\alpha}_0 - \alpha_0^\ast\|_2^2]\right\}^{1/2} \left\{\mathbb{E} \left[\sup_{\|\eta - \eta^\ast\|_2<\delta} \|\mathbb{G}_N \mathcal{F}_\eta^6\|_2^2\right]\right\}^{1/2} \\
& + N^{-1/2} \left\{\mathbb{E}[N \|\hat{\alpha}_1 - \alpha_1^\ast\|_2^2]\right\}^{1/2} \left\{\mathbb{E} \left[\sup_{\|\eta - \eta^\ast\|_2<\delta} \|\mathbb{G}_N \mathcal{F}_\eta^7\|_2^2\right]\right\}^{1/2} \\
& + \mathbb{E}\left[\sup_{\|\eta - \eta^\ast\|_2<\delta} \left| \mathbb{G}_N \mathcal{F}_\eta^8 \right| \right].
\end{align*}

Let $M_\lambda = \left\{\mathbb{E}[N \|\hat{\lambda} - \lambda^\ast\|_2^2]\right\}^{1/2}$, and $M_\theta, M_{\beta_0}, M_{\beta_1}, M_{\alpha_0}, M_{\alpha_1}$ are defined accordingly. By Condition~\ref{cond:them1}, we have that $M_\lambda, M_\theta, M_{\beta_0}, M_{\beta_1}, M_{\alpha_0}, M_{\alpha_1} < \infty$, and therefore
\begin{equation*}
(I) \leq N^{-1/2} (M_\lambda \tilde{c}_2 + M_\theta \tilde{c}_3 + M_{\beta_0} \tilde{c}_4 + M_{\beta_1} \tilde{c}_5 + M_{\alpha_0} \tilde{c}_6 + M_{\alpha_1} \tilde{c}_7) \delta^{1/2} + \tilde{c}_8 \delta^{1/2}.
\end{equation*}

In summary, we obtain that, let $N \to\infty$, the centered process satisfies
\begin{equation}\label{eq:cond2}
\begin{split}
&\mathbb{E}\left[\sqrt{N}\sup_{\|\eta - \eta^\ast\|_2 < \delta} \left|\hat{S}(t;\eta) - S(t;\eta) - \{\hat{S}(t;\eta^\ast) - S(t;\eta^\ast)\}\right| \right] \\
& \leq (I) + (II) \leq (\tilde{c}_1 + \tilde{c}_8) \delta^{1/2}.
\end{split}
\end{equation}

Let $\phi_N(\delta) = \delta^{1/2}$ and $\alpha = \frac{3}{2} < 2$, thus we have $\frac{\phi_n(\delta)}{\delta^\alpha} = \delta^{-1}$ is decreasing, and $\alpha$ does not depend on $N$. That is, the second condition holds.

{\bf Condition 3} By the facts that $\hat{\eta} \overset{p}{\to} \eta^\ast$ as $N \to \infty$, and that $\hat{S}(t;\hat{\eta}) \geq \sup_{\eta} \hat{S}(t;\eta)$, we choose $r_N = N^{1/3}$ such that $r_N^2 \phi_N(r_N^{-1}) = N^{2/3} \phi_N(N^{-1/3}) = N^{1/2}$. The third condition holds.

In the end, the three conditions are satisfied with $r_N = N^{1/3}$; thus we conclude that $N^{1/3} \|\hat{\eta} - \eta^\ast\|_2 = O_p(1)$, which completes the proof of $(iii)$ of Theorem~\ref{thm:para}.

{\bf PART 3.} We characterize the asymptotic distribution of $\hat{S}(t;\hat{\eta})$. Since we have
\begin{equation*}
\sqrt{N} \{\hat{S}(t;\hat{\eta}) - S(t;\eta^\ast)\} = \sqrt{N} \{\hat{S}(t;\hat{\eta}) - \hat{S}(t;\eta^\ast)\} + \sqrt{N}\{\hat{S}(t;\eta^\ast) - S(t;\eta^\ast)\},
\end{equation*}
we study the two terms in two steps.

{\bf Step 3.1} To establish $\sqrt{N} \{\hat{S}(t;\hat{\eta}) - \hat{S}(t;\eta^\ast)\} = o_p(1)$, it suffices to show that $\sqrt{N} \{S(t;\hat{\eta}) - S(t;\eta^\ast)\} = o_p(1)$ and $\sqrt{N} (\hat{S}(t;\hat{\eta}) - \hat{S}(t;\eta^\ast) - \{S(t;\hat{\eta}) - S(t;\eta^\ast)\}) = o_p(1)$.

First, as $N^{1/3} \|\hat{\eta} - \eta^\ast\|_2 = O_p(1)$, we take the second-order Taylor expansion 
\begin{align*}
\sqrt{N} \{S(t;\hat{\eta}) - S(t;\eta^\ast)\} & = \sqrt{N} \left\{S'(\eta^\ast) \|\hat{\eta} - \eta^\ast\|_2 + \frac{1}{2} S''(\eta^\ast) \|\hat{\eta} - \eta^\ast\|_2^2 + o_p(\|\hat{\eta} - \eta^\ast\|_2^2) \right\} \\
& = \sqrt{N} \left\{\frac{1}{2} S''(\eta^\ast)\|\hat{\eta} - \eta^\ast\|_2^2 + o_p(\|\hat{\eta} - \eta^\ast\|_2^2) \right\} \\
& = \sqrt{N} \left\{\frac{1}{2} S''(\eta^\ast) O_p(N^{-2/3}) + o_p(N^{-2/3})\right\} = o_p(1).
\end{align*}

Next, we follow the result~\eqref{eq:cond2} obtained in \textbf{PART 2}. As $N^{1/3} \|\hat{\eta} - \eta^\ast\|_2 = O_p(1)$, there exists $\tilde{\delta} = c_9 N^{-1/3}$, where $c_9 < \infty$ is a finite constant, such that $\|\hat{\eta} - \eta^\ast\|_2 \leq \tilde{\delta}$. Therefore we have
\begin{align*}
& \sqrt{N} (\hat{S}(t;\hat{\eta}) - \hat{S}(t;\eta^\ast) - \{S(t;\hat{\eta}) - S(t;\eta^\ast)\}) \\
& \leq \mathbb{E}\left[\sqrt{N} \sup_{\|\hat{\eta} - \eta^\ast\|_2 < \tilde{\delta}} \left|\hat{S}(t;\hat{\eta}) - S(t;\hat{\eta}) - \{\hat{S}(t;\eta^\ast) - S(t;\eta^\ast)\}\right| \right] \\
& \leq (\tilde{c}_1 + \tilde{c}_8) \tilde{\delta}^{1/2} = (\tilde{c}_1 + \tilde{c}_8) \sqrt{c_9} N^{-1/6} = o_p(1),
\end{align*}
which yields the result.

{\bf Step 3.2} To derive the asymptotic distribution of $\sqrt{n}\{\hat{S}(t;\eta^\ast) - S(t;\eta^\ast)\}$, we follow the result~\eqref{eq:asym.dist} obtained in \textbf{PART 1} and have that
\begin{equation*}
\sqrt{N} \left\{\hat{S}(t;\eta^\ast) - S(t;\eta^\ast)\right\} \overset{D}{\to} \mathcal{N}(0, \sigma_{t,1}^2),
\end{equation*}
where $\sigma_{t,1}^2 = \mathbb{E}[(\xi_{1,i}(t;\eta^\ast) + \xi_{2,i}(t;\eta^\ast))^2]$. Therefore we obtain in the end
\begin{align*}
\sqrt{N} \{\hat{S}(t;\hat{\eta}) - S(t;\eta^\ast)\} & = \sqrt{N} \{\hat{S}(t;\hat{\eta}) - \hat{S}(t;\eta^\ast)\} + \sqrt{N}\{\hat{S}(t;\eta^\ast) - S(t;\eta^\ast)\} \\
& = o_p(1) + \sqrt{N}\{\hat{S}(t;\eta^\ast) - S(t;\eta^\ast)\} \\
& \overset{D}{\to} \mathcal{N}(0, \sigma_{t,1}^2),
\end{align*}
which completes the proof.

For Corollary~\ref{cor:para.rmst} where we consider RMST, the proof can follow the same steps as before, and is thus omitted here.

\section{Proof of Theorem~\ref{thm:np} and Corollary~\ref{cor:np.rmst}}

Our proof has three main parts below.

{\bf PART 1.} Recall that the cross-fitting technique, at a high level as exemplified in Lemma~\ref{lem:cf}, uses sample splitting to avoid bias due to over-fitting. For simplicity, consider that the datasets $\mathcal{O}_s$ and $\mathcal{O}_t$ are randomly split into $2$ folds with equal size respectively such that $\mathcal{O}_s = \mathcal{O}_{s,1} \cup \mathcal{O}_{s,2}, \mathcal{O}_t = \mathcal{O}_{t,1} \cup \mathcal{O}_{t,2}$. The extension to $K$-folds as described in Algorithm 1 is straightforward. Here the subscript $CF$ is omitted to simplify the notation. Define $\mathcal{I}_1 = \mathcal{O}_{s,1} \cup \mathcal{O}_{t,1}, \mathcal{I}_2 = \mathcal{O}_{s,2} \cup \mathcal{O}_{t,2}$, and $N_1 = |\mathcal{I}_1|, N_2 = |\mathcal{I}_2|$. The cross-fitted estimator for the value function under the ITR $d_\eta$ is
\begin{equation*}
\hat{V}(\eta) = \frac{N_1}{N} \hat{V}^{\mathcal{I}_1}(\eta) + \frac{N_2}{N} \hat{V}^{\mathcal{I}_2}(\eta),
\end{equation*}
where
\begin{align*}
\hat{V}^{\mathcal{I}_1}(\eta) = & \frac{1}{N_1} \sum_{\mathcal{I}_1} \Bigg\{ I_{T,i} \, e(X_i) \hat{\mu}(d_\eta(X_i), X_i) + \frac{I_{S,i}}{\hat{\pi}_S(X_i)} \frac{I\{A_i = d_\eta(X_i)\}}{\hat{\pi}_{d}(X_i)} \\
& \quad \times \left(\frac{\Delta_i \, y(U_i)}{\hat{S}_C(U_i \,|\, A_i, X_i)} - \hat{\mu}(A_i, X_i) + \int_0^\infty \frac{\mathrm{d}\hat{M}_C(u \,|\, A_i, X_i)}{\hat{S}_C(u \,|\, A_i, X_i)} \hat{Q}(u, A_i, X_i) \right) \Bigg\},
\end{align*}
and the nuisance parameters are estimated from $\mathcal{I}_2$. $\hat{V}^{\mathcal{I}_2}(\eta)$ is defined accordingly.

In this step, we show that 
\begin{equation*}
\hat{V}(\eta) - V_N(\eta) = o_p(N^{-1/2}),
\end{equation*}
and essentially it suffices to prove that
\begin{equation*}
\hat{V}^{\mathcal{I}_1}(\eta) - V_N^{\mathcal{I}_1}(\eta) = o_p(N^{-1/2}),
\end{equation*}
where
\begin{align*}
V_N(\eta) = & \frac{1}{N} \sum_{i=1}^N \Bigg\{ I_{T,i} \, e(X_i) \mu(d_\eta(X_i), X_i) + \frac{I_{S,i}}{\pi_S(X_i)} \frac{I\{A_i = d_\eta(X_i)\}}{\pi_{d}(X_i)} \\
& \quad \times \left(\frac{\Delta_i \, y(U_i)}{S_C(U_i \,|\, A_i, X_i)} - \mu(A_i, X_i) + \int_0^\infty \frac{\mathrm{d}M_C(u \,|\, A_i, X_i)}{S_C(u \,|\, A_i, X_i)} Q(u, A_i, X_i) \right) \Bigg\},
\end{align*}
and $V_N^{\mathcal{I}_1}(\eta)$ is defined accordingly.

First, we have the following decomposition 
\begin{equation}\label{eq.cf.decomp}
\begin{split}
& \hat{V}^{\mathcal{I}_1}(\eta) - V_N^{\mathcal{I}_1}(\eta) \\
& = \frac{1}{N_1} \sum_{\mathcal{I}_1} \Bigg\{ I_{T,i} \, e(X_i) (\hat{\mu}(d_\eta(X_i), X_i) - \mu(d_\eta(X_i), X_i)) \\
& \quad + I_{S,i} \left(\frac{1}{\pi_S(X_i)} - \frac{1}{\hat{\pi}_S(X_i)}\right) \frac{I\{A_i = d_\eta(X_i)\}}{\pi_{d}(X_i)} K(A_i, X_i) \\
& \quad + \frac{I_{S,i} I\{A_i = d_\eta(X_i)\}}{\pi_S(X_i)} \left(\frac{1}{\pi_{d}(X_i)} - \frac{1}{\hat{\pi}_{d}(X_i)}\right) K(A_i, X_i) \\
& \quad + \frac{I_{S,i}}{\pi_S(X_i)} \frac{I\{A_i = d_\eta(X_i)\}}{\pi_{d}(X_i)} (\hat{K}(A_i, X_i) - K(A_i, X_i)) \\
& \quad + I_{S,i} I\{A_i = d_\eta(X_i)\} \left(\frac{1}{\pi_S(X_i)} - \frac{1}{\hat{\pi}_S(X_i)}\right) \left(\frac{1}{\pi_{d}(X_i)} - \frac{1}{\hat{\pi}_{d}(X_i)}\right) K(A_i, X_i) \\
& \quad + \frac{I_{S,i} I\{A_i = d_\eta(X_i)\}}{\pi_{d}(X_i)}\left(\frac{1}{\pi_S(X_i)} - \frac{1}{\hat{\pi}_S(X_i)}\right) (\hat{K}(A_i, X_i) - K(A_i, X_i)) \\
& \quad + \frac{I_{S,i} I\{A_i = d_\eta(X_i)\}}{\pi_S(X_i)} \left(\frac{1}{\pi_{d}(X_i)} - \frac{1}{\hat{\pi}_{d}(X_i)}\right) (\hat{K}(A_i, X_i) - K(A_i, X_i)) \\
& \quad + I_{S,i} I\{A_i = d_\eta(X_i)\} \left(\frac{1}{\pi_S(X_i)} - \frac{1}{\hat{\pi}_S(X_i)}\right) \left(\frac{1}{\pi_{d}(X_i)} - \frac{1}{\hat{\pi}_{d}(X_i)}\right) (\hat{K}(A_i, X_i) - K(A_i, X_i)) \Bigg\},
\end{split}
\end{equation}
where
\begin{equation*}
\hat{K}(A_i, X_i) = \frac{\Delta_i \, y(U_i)}{\hat{S}_C(U_i \,|\, A_i, X_i)} - \hat{\mu}(A_i, X_i) + \int_0^\infty \frac{\mathrm{d}\hat{M}_C(u \,|\, A_i, X_i)}{\hat{S}_C(u \,|\, A_i, X_i)} \hat{Q}(u, A_i, X_i), 
\end{equation*}
\begin{equation*}
K(A_i, X_i) = \frac{\Delta_i \, y(U_i)}{S_C(U_i \,|\, A_i, X_i)} - \mu(A_i, X_i) + \int_0^\infty \frac{\mathrm{d}M_C(u \,|\, A_i, X_i)}{S_C(u \,|\, A_i, X_i)} Q(u, A_i, X_i).
\end{equation*}

In summary, the decomposition~\eqref{eq.cf.decomp} consists of two types of terms: four mean-zero terms and four product terms. For the mean-zero terms, we utilize the method introduced in Section~\ref{subsec:prel.cf}; since 
\begin{equation*}
\mathbb{E}[I_{T,i} \, e(X_i) (\hat{\mu}(d_\eta(X_i), X_i) - \mu(d_\eta(X_i), X_i))] = 0,
\end{equation*}
by applying Lemma~\ref{lem:cf}, we obtain
\begin{equation*}
\frac{1}{N_1} \sum_{\mathcal{I}_1} I_{T,i} \, e(X_i) (\hat{\mu}(d_\eta(X_i), X_i) - \mu(d_\eta(X_i), X_i))= o_p(N^{-1/2}).
\end{equation*}

Similarly we have
\begin{equation*}
\mathbb{E} \left[I_{S,i} \left(\frac{1}{\pi_S(X_i)} - \frac{1}{\hat{\pi}_S(X_i)}\right) \frac{I\{A_i = d_\eta(X_i)\}}{\pi_{d}(X_i)} K(A_i, X_i)\right] = 0,
\end{equation*}
so we obtain
\begin{align*}
& \mathbb{E} \left[\left(\frac{1}{N_1} \sum_{\mathcal{I}_1} I_{S,i} \left(\frac{1}{\pi_S(X_i)} - \frac{1}{\hat{\pi}_S(X_i)}\right) \frac{I\{A_i = d_\eta(X_i)\}}{\pi_{d}(X_i)} K(A_i, X_i) \right)^2\right] \\
& = \mathbb{E} \left[ \mathbb{E} \left[ \left(\frac{1}{N_1} \sum_{\mathcal{I}_1} I_{S,i} \left(\frac{1}{\pi_S(X_i)} - \frac{1}{\hat{\pi}_S(X_i)}\right) \frac{I\{A_i = d_\eta(X_i)\}}{\pi_{d}(X_i)} K(A_i, X_i) \right)^2 \Bigg| \mathcal{I}_2 \right] \right] \\
& = \mathbb{E} \left[ var \left[ \frac{1}{N_1} \sum_{\mathcal{I}_1} I_{S,i} \left(\frac{1}{\pi_S(X_i)} - \frac{1}{\hat{\pi}_S(X_i)}\right) \frac{I\{A_i = d_\eta(X_i)\}}{\pi_{d}(X_i)} K(A_i, X_i) \Bigg| \mathcal{I}_2 \right] \right] \\
& = \frac{1}{N_1} \mathbb{E} \left[ var \left[ I_{S,i} \left(\frac{1}{\pi_S(X_i)} - \frac{1}{\hat{\pi}_S(X_i)}\right) \frac{I\{A_i = d_\eta(X_i)\}}{\pi_{d}(X_i)} K(A_i, X_i) \bigg| I_2 \right] \right] \\
& \leq \frac{O_p(1)}{N_1} = o_p(\frac{1}{N}).
\end{align*}

We also have
\begin{equation*}
\mathbb{E} \left[\frac{I_{S,i} I\{A_i = d_\eta(X_i)\}}{\pi_S(X_i)} \left(\frac{1}{\pi_{d}(X_i)} - \frac{1}{\hat{\pi}_{d}(X_i)}\right) K(A_i, X_i)\right] = 0,
\end{equation*}
\begin{equation*}
\mathbb{E} \left[\frac{I_{S,i}}{\pi_S(X_i)} \frac{I\{A_i = d_\eta(X_i)\}}{\pi_{d}(X_i)} (\hat{K}(A_i, X_i) - K(A_i, X_i))\right] = 0,
\end{equation*}
and using the same technique, we conclude that these two mean-zero terms are $o_p(N^{-1/2})$ as well.

The product terms can be handled simply by the Cauchy-Schwarz inequality and the rate of convergence conditions in Assumption~\ref{asmp:nonp}. Additionally we have the decomposition as follows
\begin{align*}
& \frac{1}{N_1} \sum_{\mathcal{I}_1} (\hat{K}(A_i, X_i) - K(A_i, X_i)) \\
& = \frac{1}{N_1} \sum_{\mathcal{I}_1} \Bigg\{-(\hat{\mu}(A_i, X_i) - \mu(A_i, X_i)) + \frac{1 - \Delta_i}{S_C(U_i \,|\, A_i, X_i)}(\hat{Q}(U_i \,|\, A_i, X_i) - Q(U_i \,|\, A_i, X_i)) \\
& \quad - \int_0^{U_i} \frac{\lambda_C(u \,|\, A_i, X_i)}{S_C(u \,|\, A_i, X_i)} (\hat{Q}(U_i \,|\, A_i, X_i) - Q(U_i \,|\, A_i, X_i)) \mathrm{d}u \\
& \quad + (1 - \Delta_i)\left(\frac{1}{\hat{S}_C(U_i \,|\, A_i, X_i)} - \frac{1}{S_C(U_i \,|\, A_i, X_i)}\right) Q(U_i \,|\, A_i, X_i) \\
& \quad + \left(\frac{1}{\hat{S}_C(U_i \,|\, A_i, X_i)} - \frac{1}{S_C(U_i \,|\, A_i, X_i)}\right) \Delta_i \, y(U_i) \\
& \quad - \int_0^{U_i} \left(\frac{\hat{\lambda}_C(u \,|\, A_i, X_i)}{\hat{S}_C(u \,|\, A_i, X_i)} - \frac{\lambda_C(u \,|\, A_i, X_i)}{S_C(u \,|\, A_i, X_i)} \right) Q(U_i \,|\, A_i, X_i) \mathrm{d}u \\
& \quad + (1 - \Delta_i)\left(\frac{1}{\hat{S}_C(U_i \,|\, A_i, X_i)} - \frac{1}{S_C(U_i \,|\, A_i, X_i)}\right) (\hat{Q}(U_i \,|\, A_i, X_i) - Q(U_i \,|\, A_i, X_i)) \\
& \quad - \int_0^{U_i} \left(\frac{\hat{\lambda}_C(u \,|\, A_i, X_i)}{\hat{S}_C(u \,|\, A_i, X_i)} - \frac{\lambda_C(u \,|\, A_i, X_i)}{S_C(u \,|\, A_i, X_i)} \right) (\hat{Q}(U_i \,|\, A_i, X_i) - Q(U_i \,|\, A_i, X_i)) \mathrm{d}u,
\end{align*}
and similarly we have three mean-zero terms which are $o_p(N^{-1/2})$ by the same technique in Section~\ref{subsec:prel.cf} and the facts that  
\begin{equation*}
\mathbb{E}[\hat{\mu}(A_i, X_i) - \mu(A_i, X_i)] = 0,
\end{equation*}
\begin{align*}
\mathbb{E} & \left[ \frac{1 - \Delta_i}{S_C(U_i \,|\, A_i, X_i)} (\hat{Q}(U_i \,|\, A_i, X_i) - Q(U_i \,|\, A_i, X_i)) \right. \\
& \left. \quad - \int_0^{U_i} \frac{\lambda_C(u \,|\, A_i, X_i)}{S_C(u \,|\, A_i, X_i)} (\hat{Q}(u \,|\, A_i, X_i) - Q(u \,|\, A_i, X_i)) \mathrm{d}u \right] = 0,
\end{align*}
\begin{align*}
\mathbb{E} & \bigg[(1 - \Delta_i)\left(\frac{1}{\hat{S}_C(U_i \,|\, A_i, X_i)} - \frac{1}{S_C(U_i \,|\, A_i, X_i)}\right) Q(U_i \,|\, A_i, X_i) \\
& \quad + \left(\frac{1}{\hat{S}_C(U_i \,|\, A_i, X_i)} - \frac{1}{S_C(U_i \,|\, A_i, X_i)}\right) \Delta_i \, y(U_i) \\
& \quad - \int_0^{U_i} \left(\frac{\hat{\lambda}_C(u \,|\, A_i, X_i)}{\hat{S}_C(u \,|\, A_i, X_i)} - \frac{\lambda_C(u \,|\, A_i, X_i)}{S_C(u \,|\, A_i, X_i)} \right) Q(U_i \,|\, A_i, X_i) \mathrm{d}u \bigg] = 0,
\end{align*}
and we can bound the two product terms as well
\begin{align*}
& \frac{1}{N_1} \sum_{\mathcal{I}_1} \left[ (1 - \Delta_i)\left(\frac{1}{\hat{S}_C(U_i \,|\, A_i, X_i)} - \frac{1}{S_C(U_i \,|\, A_i, X_i)}\right) (\hat{Q}(U_i \,|\, A_i, X_i) - Q(U_i \,|\, A_i, X_i)) \right.\\
& \left. \quad - \int_0^{U_i} \left(\frac{\hat{\lambda}_C(u \,|\, A_i, X_i)}{\hat{S}_C(u \,|\, A_i, X_i)} - \frac{\lambda_C(u \,|\, A_i, X_i)}{S_C(u \,|\, A_i, X_i)} \right) (\hat{Q}(U_i \,|\, A_i, X_i) - Q(U_i \,|\, A_i, X_i)) \mathrm{d}u \right] \\
& \leq \left[\frac{1}{N_1} \sum_{\mathcal{I}_1} (1 - \Delta_i)\left(\frac{1}{\hat{S}_C(U_i \,|\, A_i, X_i)} - \frac{1}{S_C(U_i \,|\, A_i, X_i)}\right)^2 \right]^{1/2} \\
& \quad \times \left[\frac{1}{N_1} \sum_{\mathcal{I}_1}(1 - \Delta_i)(\hat{Q}(U_i \,|\, A_i, X_i) - Q(U_i \,|\, A_i, X_i))^2\right]^{1/2} \\
& \quad - \int_0^{U_i} \left[\frac{1}{N_1} \sum_{\mathcal{I}_1} \left(\frac{\hat{\lambda}_C(u \,|\, A_i, X_i)}{\hat{S}_C(u \,|\, A_i, X_i)} - \frac{\lambda_C(u \,|\, A_i, X_i)}{S_C(u \,|\, A_i, X_i)} \right)^2\right]^{1/2} \\
& \quad \times \left[\frac{1}{N_1} \sum_{\mathcal{I}_1} (\hat{Q}(U_i \,|\, A_i, X_i) - Q(U_i \,|\, A_i, X_i))^2\right]^{1/2} \mathrm{d}u \\
& = o_p(N^{-1/2}),
\end{align*}
which proves that $\frac{1}{N_1} \sum_{\mathcal{I}_1} (\hat{K}(A_i, X_i) - K(A_i, X_i)) = o_p(N^{-1/2})$.

Therefore, we conclude that the four product terms in \eqref{eq.cf.decomp} are $o_p(N^{-1/2})$ as well, which completes the proof of $(i)$ in Theorem~\ref{thm:np}.

{\bf PART 2:} We show that $N^{1/3} \|\hat{\eta} - \eta^\ast\|_2 = O_p(1)$.

By Assumption~\ref{asmp:regu} $(i)$, $V(\eta)$ is twice continuously differentiable at a neighborhood of $\eta^\ast$; in \textbf{PART 1}, we show that $\hat{V}(\eta) = V(\eta) + o_p(1), \forall \eta$; since $\hat{\eta}$ maximizes $\hat{V}(\eta)$, we have that $\hat{V}(\hat{\eta}) \geq \sup_{\eta} \hat{V}(\eta)$, thus by the Argmax theorem, we have $\hat{\eta} \overset{p}{\to} \eta^\ast$ as $N \to \infty$.

In order to establish the $N^{-1/3}$ rate of convergence of $\hat{\eta}$, we apply Theorem 14.4 (Rate of convergence) of \citet{kosorok2008introduction}, and need to find the suitable rate that satisfies three conditions below.

{\bf Condition 1} For every $\eta$ in a neighborhood of $\eta^\ast$ such that $\|\eta - \eta^\ast\|_2 < \delta$, by Assumption~\ref{asmp:regu} $(i)$, we apply the second-order Taylor expansion,
\begin{align*}
V(\eta) - V(\eta^\ast) & = V'(\eta^\ast) \|\eta - \eta^\ast\|_2 + \frac{1}{2} V''(\eta^\ast) \|\eta - \eta^\ast\|_2^2 + o(\|\eta - \eta^\ast\|_2^2) \\
& = \frac{1}{2} V''(\eta^\ast)\|\eta - \eta^\ast\|_2^2 + o(\|\eta - \eta^\ast\|_2^2),
\end{align*}
and as $V''(\eta^\ast) < 0$, there exists $c_{10} = -\frac{1}{2} V''(\eta^\ast) > 0$ such that $V(\eta) - V(\eta^\ast) \leq -c_{10} \|\eta - \eta^\ast\|_2^2$. 

{\bf Condition 2} For all $N$ large enough and sufficiently small $\delta$, we consider the centered process $\hat{V} - V$, and have that 
\begin{align*}
& \mathbb{E} \left[\sqrt{N}\sup_{\|\eta - \eta^\ast\|_2 < \delta} \left|\hat{V}(\eta) - V(\eta) - \{\hat{V}(\eta^\ast) - V(\eta^\ast)\}\right| \right] \\
& = \mathbb{E}\left[\sqrt{N}\sup_{\|\eta - \eta^\ast\|_2 < \delta} \left|\hat{V}(\eta) - V_n(\eta) + V_n(\eta) - V(\eta) - \{\hat{V}(\eta^\ast) - V_n(\eta^\ast) + V_n(\eta^\ast) - V(\eta^\ast)\}\right| \right] \\
& \leq \mathbb{E}\left[\sqrt{N}\sup_{\|\eta - \eta^\ast\|_2 < \delta} \left|\hat{V}(\eta) - V_n(\eta) - \{\hat{V}(\eta^\ast) - V_n(\eta^\ast)\}\right| \right] \tag{$I$} \\
& \quad + \mathbb{E}\left[\sqrt{N}\sup_{\|\eta - \eta^\ast\|_2 < \delta} \left|V_n(\eta) - V(\eta) - \{V_n(\eta^\ast) - V(\eta^\ast)\}\right| \right] \tag{$II$}
\end{align*}

It follows from the result in \textbf{PART 1} that $(I) = o_p(1)$. To bound $(II)$, we have
\begin{align*}
& V_n(\eta) - V_n(\eta^\ast) \\
& = \frac{1}{N} \sum_{i=1}^{N} (d_{\eta}(X_i) - d_{\eta^\ast}(X_i)) \times \left( I_{T,i} \, e(X_i) (\mu(1, X_i) - \mu(0, X_i)) + \frac{(2A_i - 1) I_{S,i}}{\pi_{A_i}(X_i)\pi_S(X_i)} K(A_i, X_i) \right).
\end{align*}

Define a class of functions
\begin{equation*}
\mathcal{F}_\eta^9 = \bigg\{ (d_{\eta}(x) - d_{\eta^\ast}(x)) \times \bigg(I_{T} \, e(x) (\mu(1, x) - \mu(0, x)) + \frac{(2a - 1) I_{S}}{\pi_{a}(x)\pi_S(x)} K(a, x) \bigg) :\|\eta - \eta^\ast\|_2 < \delta \bigg\},
\end{equation*}
and let $M_9 = \sup\left|I_{T} \, e(x) (\mu(1, x) - \mu(0, x)) + \frac{(2a - 1) I_{S}}{\pi_{a}(x)\pi_S(x)} K(a, x)\right|$. By Assumption~\ref{asmp:cnpc}, \ref{asmp:psi} and Condition~\ref{cond:them1}, we have that $M_9 < \infty$. Using the same technique as in Section~\ref{subsec:asymp} {\bf Condition 2.1}, we define the envelop of $\mathcal{F}_\eta^9$ as $F_9 = M_9 I\{-k_0 \delta \leq (1, x^T) \eta^\ast \leq k_0 \delta\}$, and obtain that $\|F_9\|_{P,2} \leq \tilde{M}_9 \delta^{1/2} < \infty$, where $\tilde{M}_9$ is a finite constant, and that $\mathcal{F}_\eta^9$ is a VC class with bounded entropy $J_{[]}^\ast(1,\mathcal{F}_\eta^9) < \infty$. By Theorem 11.2 of \cite{kosorok2008introduction}, we obtain
\begin{equation*}
\mathbb{E}\left[\sup_{\|\eta - \eta^\ast\|_2<\delta} \left|\mathbb{G}_N \mathcal{F}_\eta^9\right| \right] \leq \tilde{c}_9 \delta^{1/2},
\end{equation*}
where $\tilde{c}_9$ is a finite constant. Therefore, we obtain
\begin{equation*}
\begin{split}
(II) & = \mathbb{E}\left[\sqrt{N}\sup_{\|\eta - \eta^\ast\|_2 < \delta} \left|V_N(\eta) - V(\eta) - \{V_N(\eta^\ast) - V(\eta^\ast)\} \right| \right] \\
&= \mathbb{E}\left[\sup_{\|\eta - \eta^\ast\|_2 < \delta} |\mathbb{G}_n\mathcal{F}_\eta^9|\right] \leq \tilde{c}_9 \delta^{1/2}.
\end{split}
\end{equation*}

In summary, we obtain that the centered process satisfies
\begin{equation}\label{eq:cond2.thm2}
\begin{split}
&\mathbb{E}\left[\sqrt{N}\sup_{\|\eta - \eta^\ast\|_2 < \delta} \left|\hat{S}(t;\eta) - S(t;\eta) - \{\hat{S}(t;\eta^\ast) - S(t;\eta^\ast)\}\right| \right] \\
& \leq (I) + (II) \leq \tilde{c}_9 \delta^{1/2}.
\end{split}
\end{equation}

Let $\phi_N(\delta) = \delta^{1/2}$ and $\alpha = \frac{3}{2} < 2$, thus we have $\frac{\phi_n(\delta)}{\delta^\alpha} = \delta^{-1}$ is decreasing, and $\alpha$ does not depend on $N$. That is, the second condition holds.

{\bf Condition 3} By the facts that $\hat{\eta} \overset{p}{\to} \eta^\ast$ as $N \to \infty$, and that $\hat{S}(t;\hat{\eta}) \geq \sup_{\eta} \hat{S}(t;\eta)$, we choose $r_N = N^{1/3}$ such that $r_N^2 \phi_N(r_N^{-1}) = N^{2/3} \phi_N(N^{-1/3}) = N^{1/2}$. The third condition holds.

In the end, the three conditions are satisfied with $r_N = N^{1/3}$; thus we conclude that $N^{1/3} \|\hat{\eta} - \eta^\ast\|_2 = O_p(1)$, which completes the proof of $(ii)$ in Theorem~\ref{thm:np}.

{\bf PART 3:} We characterize the asymptotic distribution of $\hat{V}(\hat{\eta})$. Since we have
\begin{equation*}
\sqrt{N} \{\hat{V}(\hat{\eta}) - V(\eta^\ast)\} = \sqrt{N} \{\hat{V}(\hat{\eta}) - \hat{V}(\eta^\ast)\} + \sqrt{N}\{\hat{V}(\eta^\ast) - V(t;\eta^\ast)\},
\end{equation*}
we study the two terms in two steps.

{\bf Step 3.1} To establish $\sqrt{N} \{\hat{V}(\hat{\eta}) - \hat{V}(\eta^\ast)\} = o_p(1)$, it suffices to show that $\sqrt{N} \{V(\hat{\eta}) - V(\eta^\ast)\} = o_p(1)$ and $\sqrt{N} (\hat{V}(\hat{\eta}) - \hat{V}(\eta^\ast) - \{V(\hat{\eta}) - V(\eta^\ast)\}) = o_p(1)$.

First, as $N^{1/3} \|\hat{\eta} - \eta^\ast\|_2 = O_p(1)$, we take the second-order Taylor expansion 
\begin{align*}
\sqrt{N} \{V(\hat{\eta}) - V(\eta^\ast)\} & = \sqrt{N} \left\{V'(\eta^\ast) \|\hat{\eta} - \eta^\ast\|_2 + \frac{1}{2} V''(\eta^\ast) \|\hat{\eta} - \eta^\ast\|_2^2 + o_p(\|\hat{\eta} - \eta^\ast\|_2^2) \right\} \\
& = \sqrt{N} \left\{\frac{1}{2} V''(\eta^\ast)\|\hat{\eta} - \eta^\ast\|_2^2 + o_p(\|\hat{\eta} - \eta^\ast\|_2^2) \right\} \\
& = \sqrt{N} \left\{\frac{1}{2} V''(\eta^\ast) O_p(N^{-2/3}) + o_p(N^{-2/3})\right\} = o_p(1).
\end{align*}

Next, we follow the result~\eqref{eq:cond2.thm2} obtained in \textbf{PART 2}. As $N^{1/3} \|\hat{\eta} - \eta^\ast\|_2 = O_p(1)$, there exists $\tilde{\delta}_2 = c_{11} N^{-1/3}$, where $c_{11} < \infty$ is a finite constant, such that $\|\hat{\eta} - \eta^\ast\|_2 \leq \tilde{\delta}_2$. Therefore we have
\begin{align*}
& \sqrt{N} (\hat{V}(\hat{\eta}) - \hat{V}(\eta^\ast) - \{V(\hat{\eta}) - V(\eta^\ast)\}) \\
& \leq \mathbb{E}\left[\sqrt{N} \sup_{\|\hat{\eta} - \eta^\ast\|_2 < \tilde{\delta}_2} \left|\hat{V}(\hat{\eta}) - V(\hat{\eta}) - \{\hat{V}(\eta^\ast) - V(\eta^\ast)\}\right| \right] \\
& \leq \tilde{c}_9 \tilde{\delta}^{1/2} = \tilde{c}_9 \sqrt{c_{11}} N^{-1/6} = o_p(1),
\end{align*}
which yields the result.

{\bf Step 3.2} To derive the asymptotic distribution of $\sqrt{N}\{\hat{V}(\eta^\ast) - V(\eta^\ast)\}$, we follow the result obtained in \textbf{PART 1} that $\hat{V}(\eta^\ast) = V_N(\eta^\ast) + o_p(N^{-1/2})$, and thus
\begin{equation*}
\sqrt{N} \left\{\hat{V}(\eta^\ast) - V(\eta^\ast)\right\} \overset{D}{\to} \mathcal{N}(0, \sigma_{2}^2),
\end{equation*}
where $\sigma_{2}^2 = \mathbb{E}[\phi_{d_{\eta^\ast}}^2]$ is the semiparametric efficiency bound. 

Therefore we obtain in the end
\begin{align*}
\sqrt{N} \{\hat{V}(\hat{\eta}) - v(\eta^\ast)\} & = \sqrt{N} \{\hat{V}(\hat{\eta}) - \hat{V}(\eta^\ast)\} + \sqrt{N}\{\hat{V}(\eta^\ast) - V(\eta^\ast)\} \\
& = o_p(1) + \sqrt{N}\{\hat{V}(\eta^\ast) - V(\eta^\ast)\} \\
& \overset{D}{\to} \mathcal{N}(0, \sigma_{2}^2),
\end{align*}
which completes the proof of Theorem~\ref{thm:np} and Corollary~\ref{cor:np.rmst}.

\section{Proof of Theorem~\ref{thm:st} and Theorem~\ref{thm:tl}}

When the source and target populations have the same distributions, both $\hat{V}_{DR} (\eta)$ and $\hat{V}_{CF} (\eta)$ converge to $V(\eta)$. The asymptotic variance of $\hat{V}_{DR} (\eta)$ is
\begin{align*}
\sigma_{DR}^2 & = \mathbb{E} \left[\frac{I_S}{\mathbb{P}(I_S = 1)} \left(\mu(d(X), X) + \frac{I\{A = d(X)\}}{\pi_d (X)} K(A, X) - V(\eta)\right)^2 \right] \\
& = \mathbb{E} \left[\frac{I_S}{\mathbb{P}(I_S = 1)} \left(\mu^2(d(X), X) + \frac{I\{A = d(X)\}}{\pi_d^2 (X)} K^2(A, X) - V^2(\eta) \right.\right. \\
& \qquad\quad \left.\left. + \frac{2 I\{A = d(X)\}}{\pi_d (X)} K(A, X)\mu(d(X), X) - 2 \mu(d(X), X)V(\eta) \right.\right. \\
&\qquad\quad \left.\left. - \frac{2 I\{A = d(X)\}}{\pi_d (X)} K(A, X)V(\eta) \right) \right],
\end{align*}
while the asymptotic variance of $\hat{V}_{CF} (\eta)$ is
\begin{align*}
\sigma_{CF}^2 & = \mathbb{E} \left[\left(I_T \, e(X) \mu(d(X), X) + \frac{I_S \, I\{A = d(X)\}}{\pi_S (X) \pi_d (X)} K(A, X) - V(\eta)\right)^2\right] \\
& = \mathbb{E} \left[\left(I_T \, e^2(X)\mu^2(d(X), X) + \frac{I_S \, I\{A = d(X)\}}{\pi_S^2 (X) \pi_d^2 (X)} K^2(A, X) - V^2(\eta) \right.\right. \\
& \qquad\quad \left.\left. - 2 I_T \, e^2(X) \mu(d(X), X)V(\eta) - \frac{2 I_S \, I\{A = d(X)\}}{\pi_S (X) \pi_d (X)} K(A, X)V(\eta) \right) \right],
\end{align*}
where
\begin{equation*}
K(A, X) = \frac{\Delta \, y(U)}{S_C(U \,|\, A, X)} - \mu(A, X) + \int_0^\infty \frac{\mathrm{d}M_C(u \,|\, A, X)}{S_C(u \,|\, A, X)} Q(u, A, X).
\end{equation*}

Since we have that
\begin{equation*}
\mathbb{E} \left[\frac{I_S}{\mathbb{P}(I_S = 1)} \frac{2 I\{A = d(X)\}}{\pi_d (X)} K(A, X)\mu(d(X), X) \right] = 0,   
\end{equation*}
and for
\begin{equation*}
B \in \left\{\mu^2(d(X), X), \frac{I\{A = d(X)\}}{\pi_d^2 (X)} K^2(A, X), \mu(d(X), X)V(\eta), \frac{I\{A = d(X)\}}{\pi_d^2 (X)} K(A, X)V(\eta) \right\},
\end{equation*}
we also have that 
\begin{equation*}
\mathbb{E} \left[\frac{I_S}{\mathbb{P}(I_S = 1)} B \right] = \mathbb{E} [I_T \, e(X) B] = \mathbb{E} \left[\frac{I_S}{\pi_S (X)} B\right],
\end{equation*}
we conclude that $\sigma_{DR}^2 = \sigma_{CF}^2$.

By the law of iterated expectations, the value function $V_d = \mathbb{E} [y(T(d))] = \mathbb{E}_X [\mathbb{E} [y(T(d)) \,|\, X]]$. When there is no restriction on the class of ITRs, the true optimal ITR is
\begin{align*}
d^{\ast\ast} (X) & = \arg\max_d V_d = \arg\max_d \mathbb{E}_X [\mathbb{E} [y(T(d)) \,|\, X]] \\
& = I\{\mathbb{E} [y(T(1)) \,|\, X] > \mathbb{E} [y(T(0)) \,|\, X]\}.
\end{align*}
That is, the optimal ITR does not depend on the covariate distributions, but only the bilp function which is the same in both the source and target populations by Assumption~\ref{asmp:sme}. Thus both the maximizers of $\hat{V}_{DR} (\eta)$ and $\hat{V}_{CF} (\eta)$ converge to the true population parameter $\eta^{\ast\ast}$. However, $\hat{V}_{DR} (\eta)$ is biased since the expectation $\mathbb{E}_X$ is taken with respect to the source population.

\section{Additional simulations}

We first investigate the performance of the cross-fitted ACW estimator with different sample sizes $(N, m) = (5 \times 10^4, 2000), (1 \times 10^5, 4000), (2 \times 10^5, 8000), (4 \times 10^5, 16000), (6 \times 10^5, 24000), (8 \times 10^5, 32000)$. Figure~\ref{fig:simu.np1} and Table~\ref{tab:simu.np1} report the results from $200$ Monte Carlo replications. The variance is computed using the EIF.

\begin{figure}[p]
    \centering
    \caption{Boxplot of estimated value by ACW estimator with different sample sizes.}
    \includegraphics{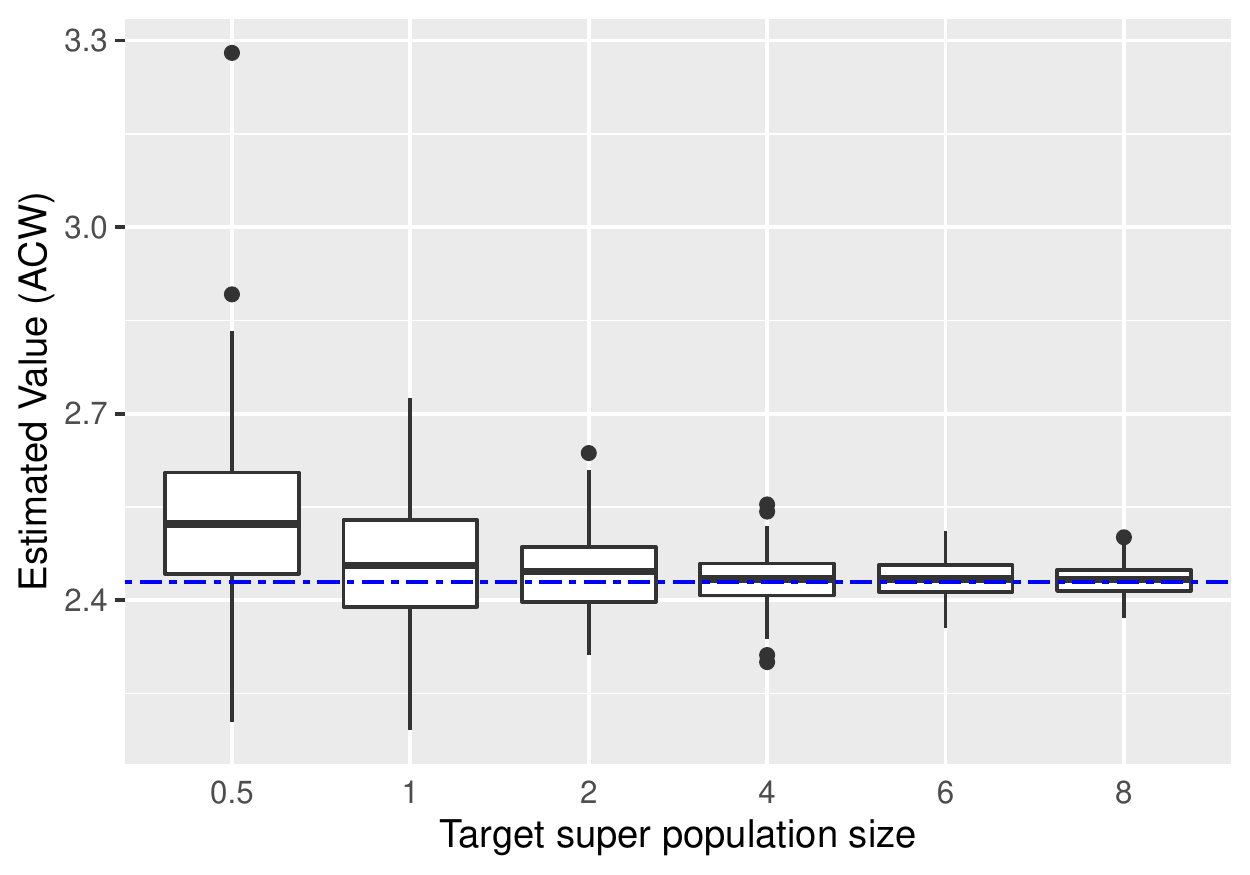}
    \label{fig:simu.np1}
\end{figure}

\begin{table}[p]
\centering
\caption{Numeric results of the ACW estimator. Bias is the empirical bias of point estimates; SD is the empirical standard deviation of point estimates; SE is the average of standard error estimates; CP is the empirical coverage probability of the 95\% Wald confidence intervals.} \bigskip
\begin{tabular}{l cccccc}
\hline\hline
$n; m (\times 10^3)$ & $\sim 780; 2$ & $\sim 1560; 4$ & $\sim 3120; 8$ & $\sim 6240; 16$ & $\sim 9360; 24$ & $\sim 12480; 32$ \\
\hline
Bias & 0.1041 & 0.0253 & 0.0134 & 0.0046 & 0.0031 & 0.0030 \\
SD & 0.1394 & 0.0985 & 0.0635 & 0.0419 & 0.0317 & 0.0267 \\
SE & 0.1611 & 0.0942 & 0.0627 & 0.0417 & 0.0330 & 0.0284 \\
CP(\%) & 97.5 & 93.5 & 96.0 &94.5 & 97.5 & 97.0 \\
\hline\hline
\end{tabular}
\label{tab:simu.np1}
\end{table}

\section{Details of real data analysis}

There are around $0.5\%$ and $1.6\%$ missing values in the RCT and OS data, respectively. We use the \texttt{mice} function in the \texttt{R} package \texttt{mice} \citep{van2011mice} to impute the missing values. 

Motivated by the clinical practice and existing work in the medical literature, we consider ITRs that depend on the following five variables:

\begin{itemize}
    \item AGE, SEX and Sequential Organ Failure Assessment (SOFA) score: these three baseline variables are well related to mortality in ICUs, so we consider them as important risk factors.
    \item Acute Kidney Injury Network (AKIN) score: \cite{jaber2018sodium} observed that the infusion of sodium bicarbonate improved survival outcomes and mortality rate in critically ill patients with severe metabolic acidemia and acute kidney injury. In the observational data, the AKIN score was not recorded, so we computed the score using serum creatinine measurement \citep{zavada2010comparison}.
    \item SEPSIS: we consider the presence of sepsis as a risk factor because it is the main condition associated with severe acidemia at the arrival in ICU. The effect of sodium bicarbonate infusion on patients with acidema and acute kidney injury was also observed in septic patients \citep{zhang2018effectiveness}.
\end{itemize}

\end{document}